\documentclass[12pt]{article}
\usepackage{amsmath}
\usepackage{graphicx}
\usepackage{enumerate}
\usepackage[numbers]{natbib}

\usepackage{url} %
\usepackage{caption,setspace}
\usepackage{amssymb, amsmath, mathtools, amsthm} %
\usepackage{bbm}
\usepackage{multirow}
\usepackage[linesnumbered,ruled,vlined]{algorithm2e}
\usepackage[noend]{algpseudocode}
\usepackage{microtype}
\usepackage{booktabs}
\usepackage[ colorlinks = true,
            linkcolor = blue,
            urlcolor  = blue,
            citecolor = blue,
            anchorcolor = blue]{hyperref}
\usepackage[dvipsnames]{xcolor}           
\usepackage[all]{hypcap}
\hypersetup{citecolor=Violet}
\hypersetup{linkcolor=black}
\hypersetup{urlcolor=MidnightBlue}
\usepackage{mathtools}
\usepackage{url}
\usepackage[utf8]{inputenc}
\usepackage{subfig}
\usepackage{indentfirst}
\usepackage{framed}

\usepackage[bitstream-charter]{mathdesign}
\usepackage{amsmath}
\usepackage[scaled=0.92]{PTSans}

\usepackage[nameinlink]{cleveref}

\usepackage[acronyms,smallcaps,nowarn]{glossaries-extra}
\glsdisablehyper{}

\crefname{figure}{fig.}{figs.}
\Crefname{figure}{Fig.}{Figs.}
\crefname{equation}{eq.}{eqs.}
\Crefname{equation}{Eq.}{Eqs.}
\crefname{algorithm}{alg.}{algs.}
\Crefname{algorithm}{Alg.}{Algs.}
\Crefname{definition}{Def.}{Def.}
\crefname{section}{\S}{\S\S}
\Crefname{section}{\S}{\S\S}
\crefname{lemma}{lem.}{lems.}
\Crefname{lemma}{Lem.}{Lems.}
\crefname{theorem}{thm.}{thms.}
\Crefname{theorem}{Thm.}{Thms.}
\crefname{proposition}{prop.}{props.}
\Crefname{proposition}{Prop.}{Props.}

\newtheorem{theorem}{Theorem}
\newtheorem{definition}{Definition}
\newtheorem{lemma}{Lemma}

\newtheorem{proposition}{Proposition}
\newtheorem{assumption}{Assumption}

\newcommand{\half}{\frac{1}{2}}

\newcommand{\bas}[1]{\begin{align*}#1\end{align*}}
\newcommand{\bac}[1]{\begin{equation}\begin{aligned}#1\end{aligned}\end{equation}}
\newcommand{\ba}[1]{\begin{align}#1\end{align}}

\newcommand{\norm}[1]{\left\lVert#1\right\rVert}

\newcommand{\idx}[1]{\{1,2,\cdots,#1\}}

\newcommand{\distas}[1]{\mathbin{\overset{#1}{\kern\z@\sim}}}

\newcommand{\indep}{\rotatebox[origin=c]{90}{$\models$}}

\newcommand{\g}{\,\vert\,}

\newcommand{\cN}{\mathcal{N}}

\newcommand{\wv}{\boldsymbol{w}}

\newcommand{\bR}{\mathbb{R}}
\newcommand{\E}{\mathbb{E}}
\newcommand{\cE}{\mathcal{E}}

\usepackage{float}
\usepackage{layouts}
\usepackage{placeins}
\newcommand{\Z}{\mathcal{Z}}
\newcommand{\I}{\mathcal{I}}
\newcommand{\bP}{\mathbb{P}}

\algdef{SE}[DOWHILE]{Do}{doWhile}{\algorithmicdo}[1]{\algorithmicwhile\ #1}%
\def\NoNumber#1{{\def\alglinenumber##1{}\State #1}\addtocounter{ALG@line}{-1}}
\newenvironment{proof-sketch}{\noindent{\it Proof Sketch}\hspace*{1em}}{\qed\bigskip}

\usepackage{inputenc}
\usepackage{newunicodechar}

\usepackage{titlesec}
\titleformat{\subsubsection}[runin]{\normalfont\bfseries}{\thesubsubsection}{1em}{}

\newtheoremstyle{bfnoteonly}%
{}{}%
{\itshape}{}%
{}{.}%
{ }%
{\thmnote{#3}}

\let\oldnl\nl%
\newcommand{\nonl}{\renewcommand{\nl}{\let\nl\oldnl}}%

\theoremstyle{bfnoteonly}

\usepackage{array}
\usepackage{multirow}

\newcommand{\blind}{1}

\SetAlFnt{\small}
\begin{document}

\def\spacingset#1{\renewcommand{\baselinestretch}%
{#1}\small\normalsize} \spacingset{1}

\if1\blind
{
  \title{\bf Conformal Sensitivity Analysis for \\Individual Treatment Effects}
\author{
Mingzhang Yin \thanks{Columbia University, Data Science Institute and Irving Institute for Cancer Dynamics; 
{Correspondence to:}~\url{mzyin11@gmail.com}.
}
\and  Claudia Shi \thanks{Columbia University, Department of Computer Science;
{\em email:}~\url{claudia.j.shi@gmail.com}.
}
\and Yixin Wang  \thanks{University of Michigan, Department of Statistics;
{\em email:}~\url{yixinw@umich.edu}.
}
\and David M. Blei \thanks{Columbia University, Department of Computer Science and Department of Statistics;
{\em email:}~\url{david.blei@columbia.edu}.
}
}
  \maketitle
} 
\fi

\if0\blind
{
  \bigskip
  \bigskip
  \bigskip
  \begin{center}
    {\LARGE\bf Conformal Sensitivity Analysis for \\ Individual Treatment Effects}
\end{center}
  \medskip
} \fi

\bigskip
\begin{abstract}
  Estimating an individual treatment effect (ITE) is essential to
  personalized decision making.  However, existing methods for
  estimating the ITE often rely on unconfoundedness, an assumption
  that is fundamentally untestable with observed data.  To assess the robustness of individual-level causal conclusion with unconfoundedness,
  this paper proposes a method for sensitivity analysis of the ITE, a
  way to estimate a range of the ITE under unobserved confounding.
  The method we develop quantifies unmeasured confounding through a
  marginal sensitivity
  model~\citep{rosenbaumn2002observational,tan2006distributional}, and
   adapts the framework of conformal inference to estimate an ITE
  interval at a given confounding strength.  In particular, we
  formulate this sensitivity analysis  as a conformal
  inference problem under distribution shift, and we extend existing methods
  of covariate-shifted conformal inference to this more general
  setting.  The resulting predictive interval  has guaranteed
  nominal coverage of the ITE and provides this coverage with
  distribution-free and nonasymptotic guarantees.  We evaluate the
  method on synthetic data and illustrate its application in an
  observational study.
\end{abstract}

\noindent%
\vfill

\newpage
\spacingset{1.9} %
\addtolength{\textheight}{.5in}%
\addtolength{\textheight}{-.3in}%
\section{Introduction}
Consider a person who ponders whether to take the COVID-19 vaccine.  She is
interested in understanding how much risk the COVID vaccine can reduce
for her. However, most large-scale observational studies are conducted
to estimate the average vaccine efficacy over a whole population
\citep{haas2021impact}, and such population-level estimates provide a
summary that cannot reflect individual heterogeneity. %

The causal estimand that captures individual heterogeneity is the
individual treatment effect (ITE), the per-individual difference
between the potential outcomes. However, estimation of the ITE is
fundamentally challenging, even beyond the usual population-level
estimands, because of its inherent uncertainty. To address this
challenge, researchers have recently adapted the method of conformal
inference~\citep{vovk2005algorithmic} to estimate ITE intervals with
good theoretical guarantees
\citep{kivaranovic2020conformal,lei2020conformal}. Conformal inference
helps estimate an interval that contains the true ITE with a
guaranteed minimal probability.

Conformal inference for ITE estimation is an important innovation, but
it comes with assumptions. In particular, it relies on the usual
assumption of unconfoundedness \citep{kivaranovic2020conformal,lei2020conformal} that the treatment assignment is
conditionally independent of the potential outcomes. In practice, this
assumption can be difficult to accept for many observational
studies~\citep{greenland1999confounding}, and violations of
unconfoundedness will introduce hidden biases into the estimation of
the ITE. In the context of COVID-19 vaccine studies, for example,
unmeasured confounding may come from
coexisting conditions, medical resources, and socioeconomic status
\citep{amin2021causation}.

To assess the robustness of individual-level causal conclusions with unconfoundedness, this paper develops a method for sensitivity analysis of
the ITE. %
The idea of sensitivity
analysis is to  quantify the violation of the required
assumptions and then to produce intervals of causal estimates that
account for such violations. In the context of the ITE, a sensitivity
analysis must account for two sources of uncertainty:  the
inherent uncertainty of the estimand itself and the uncertainty due to
violations of the required assumptions.

We develop conformal sensitivity analysis (CSA), a method for
sensitivity analysis of ITE interval estimation. Given a pre-specified
amount of unmeasured confounding, CSA estimates an interval that
captures the true ITE with a guaranteed probability. We develop CSA by
relaxing the assumption of unconfoundedness with a marginal
sensitivity model
(MSM)~\citep{rosenbaumn2002observational,tan2006distributional}, a
general model of the treatment assignment and potential outcomes that
includes a real-valued parameter for the strength of unmeasured
confounding. With the MSM in hand, we then show how sensitivity
analysis can be formulated as a predictive inference of the missing
potential outcomes, but under a general distribution shift. Finally,
we extend weighted conformal prediction~\citep{tib2019conformal}, a
predictive inference method developed in the setting of covariate
shift, to this more general setting of distribution shift.

The  CSA algorithm contains two stages. Given a specification
of an MSM, it first computes the range of weight functions of the
weighted conformal prediction. While covariate shift leads to a single
weight function, distribution shift requires a range. Then it uses
these functions to quantify the bounds of the ITE, found by solving an
optimization problem with constrained weights. %
The resulting  algorithm provides a valid interval estimate of an ITE whenever the true data generation is consistent with the MSM.

CSA has several practical and theoretical strengths. By leveraging
conformal inference, it makes minimal assumptions about the underlying
distribution of the observed data, and its theoretical guarantees are
valid with finite data. By using an MSM, it does not impose
additional untestable assumptions over the distribution of a latent confounder and its effects on other variables. CSA can be used with any
predictive functions to fit the treatment and outcome, and it can be
applied after fitting such functions with a light computational cost.

\subsection{Related Work}

\noindent \textbf{Conformal Inference.~~~} 
 The framework of conformal inference %
 was pioneered by Vladimir Vovk and his collaborators
 \citep{papadopoulos2002inductive,vovk2012conditional,vovk2005algorithmic,vovk2009line}.
Recent developments of conformal inference improve its  accuracy
 \citep{lei2018distribution,romano2019conformalized}, efficiency
 \citep{lei2015conformal},
 and extend its applicable domains
\citep{candes2021conformalized,lei2015conformal}.

First, regarding accuracy, a variety of conformal inference algorithms were proposed to reduce the length of predictive band. Some algorithms rely on the conditional quantile regression of the outcome given the covariates to capture the individual heterogeneity \citep{kivaranovic2020adaptive,romano2019conformalized,sesia2020comparison}, some adapt to skewed data by estimating the conditional histograms \citep{Sesia2021-ei}, and others  estimate the conditional density function to produce nonconvex predictive bands \citep{Hoff2021-hk,izbicki2020flexible}. Second, to improve efficiency, the split conformal inference framework is proposed; it uses data splitting to avoid multiple
 re-fitting of the predictor \citep{lei2014distribution,papadopoulos2008inductive,shafer2008tutorial}. Such data-splitting will also be adopted in this paper. Finally,  regarding domain extensions,  the weighted conformal prediction is proposed to handle non-i.i.d. data \citep{tib2019conformal}, generalizing conformal inference from exchangeable data to
  data with covariates
 shift.

\noindent \textbf{Sensitivity
   analysis.~~~}
 Sensitivity analysis dates back to the study of the average treatment effect (ATE)
 of smoking on lung cancer  \citep{cornfield1959smoking}.  %
More recent advances for sensitivity analysis  posit  a hypothetical latent confounder and evaluate its impact on a causal conclusion \citep{cinelli2020making, ding2016sensitivity, dorie2016flexible,hong2021did,imbens2003sensitivity,rosenbaum1983assessing,veitch2020sense}. Though intuitive, introducing a latent confounder often entails
 additional untestable assumptions
 \citep{franks2019flexible}. As an alternative, some methods
 directly model
 the dependency between treatment assignment and potential outcomes
 given the covariates, such as the MSM in this paper \citep{robins2000sensitivity,tan2006distributional}. With this strategy, some sensitivity models focus on modeling the potential
 outcome given the treatment
 \citep{blackwell2014selection,brumback2004sensitivity}, while others
 focus on modeling the treatment distribution given the potential outcomes
 \citep{franks2019flexible,tan2006distributional,yadlowsky2018bounds,zhao2019sensitivity}.

Other papers consider different frameworks to evaluate the sensitivity of a causal estimate.  Some sensitivity analysis methods measure the association between a latent confounder and the treatment or outcome that produces a specific amount of estimation bias \citep{cinelli2020making,ding2016sensitivity,imbens2003sensitivity,veitch2020sense}. %
Another way to evaluate sensitivity is to compute an interval estimate of the target estimand for a specific level of unmeasured confounding.   For the ATE, the percentile bootstrap 
 produces a partial identified region with asymptotically valid coverage \citep{zhao2019sensitivity}.
 For the conditional average treatment effect (CATE), data-dependent interval estimations have been proposed
 via nonparametric and (semi-)parametric approaches
 \citep{jesson2021quantifying,kallus2019interval,yadlowsky2018bounds}. This work further explores the interval estimation of the ITE under an unmeasured confounding.

In an independent and concurrent paper, \citet{2021jin} also develops sensitivity analysis procedures for the ITE based on robust conformal inference. \citet{2021jin} derives a sensitivity analysis based on the MSM and proposes an extended conformal inference algorithm that is equivalent to \Cref{alg:csa}. 
However, this paper and   \citet{2021jin}  propose the methods of analysis that are complementary and offer different perspectives. The present  paper defines the MSM without explicitly having to posit a latent confounder and uses Tukey’s factorization for an alternative derivation of \Cref{lm:sscore} (same as Lem. 3.1 \citep{2021jin}). We propose and implement an algorithm  to improve the sharpness of the predictive set, provide tools of calibration from observed data, and design methods to evaluate the estimation over different sensitivity models in the MSM.  ~\looseness=-1

\vspace{-2mm}
\section{Conformal Inference of Individual Treatment Effects} \label{sec:setup}
We first set up the problem of  ITE estimation. Next, we formulate the ITE estimation in observational study as a conformal inference problem under distribution shift, and introduce existing estimation methods under the assumption of unconfoundedness. Then we discuss the challenges presented to ITE estimation when there is unmeasured confounding.

\subsection{Problem Setup}
Consider $N$ statistical units. Each unit $i \in \idx{N}$ is associated with a tuple of random variables $(X_i, T_i, Y_i(0), Y_i(1))$. 
$X_i \in \mathcal{X} \subset \bR^p$ is a vector of covariates,
$T_i \in \{0,1\}$ is the treatment,
$Y_i(1), Y_i(0)  \in \mathcal{Y} \subset \bR$ are the potential outcomes under treatment and control \citep{neyman1923applications,rubin1974estimating}.
We use $\mathbb{P}_0(X, T, Y(0), Y(1))$ to denote the true joint distribution of these variables.

We  make the stable unit treatment value assumption (SUTVA) \citep{rubin1980randomization}. Under SUTVA, the observed outcome $Y_i \in \bR$ is one of the potential outcomes $Y_i = T_iY_i(1) + (1-T_i)Y_i(0)$. 

\begin{assumption}[SUTVA] 
There is no interference between units, and there are no unrepresented treatments.
\end{assumption} 
We further assume that each unit has a positive probability of being assigned to all treatment groups and the probability is bounded away from the extremes \citep{rosenbaum1983central}.
\begin{assumption}[Strong overlap]    $\exists~\eta>0$,
$\eta<p(T_i=1\g X_i)<1-\eta$ with probability 1.     \label{assp:overlap}           
\end{assumption} 

The causal estimand of interest  is the ITE. The ITE of unit $i$ is defined as the difference between its potential outcomes, $\tau_i = Y_i(1) - Y_i(0)$.

Estimating the ITE is challenging.
The fundamental problem of causal inference is that we can at most observe one potential outcome of each unit  \citep{holland1986statistics}.
Therefore, the ITE, which requires knowing both the potential outcomes, can never be observed.
Furthermore, unlike population-level causal estimands, an ITE is inherently random.
Even with a known joint distribution $\mathbb{P}_0(X, T, Y(0), Y(1))$,  an ITE is not point-identifiable \citep{hernan2010causal}.

To tackle these challenges,  the problem of ITE estimation has been re-framed as a predictive inference problem  \citep{kivaranovic2020conformal,lei2020conformal}.
\subsection{Predictive Inference in Observational Studies}\label{subsec:predictive}
The idea of predictive inference is to form a covariate-dependent predictive band that contains the outcome of a new data point with a guaranteed probability \citep{barber2021predictive,vovk2005algorithmic}.  
For predicting potential outcomes, predictive inference aims to use observed data $\{X_i, Y_i(t)\}_{i:T_i=t}$ from the treatment group $t$  to learn a mapping from the covariates $X$ to an interval estimate $\hat{C}_t(X) \subset \bR$ of the potential outcome $Y(t)$. 
For a new data point $(X,Y(t)) \sim \mathbb{P}(X,Y(t))$, the band must have a valid coverage probability $\mathbb{P}(Y(t) \in \hat{C}_t(X)) \geq 1-\alpha$, where the probability is taken over both $X$ and $Y(t)$ and $\alpha \in [0,1]$ is a pre-determined level. %

Conformal inference is a collection of methods that realize the goal of predictive inference \citep{vovk2005algorithmic}. 
Classic conformal inference assumes the training data and the target data are \emph{exchangeable}. 
 Based on exchangeability,  the predictive band is constructed by the quantiles of the prediction residuals.
Weighted conformal prediction (WCP) extends the setting to  covariate shift  \citep{tib2019conformal}, where distribution of the covariates $\mathbb{P}(X)$ changes from the training data to  the target data  but the  outcome distribution $\mathbb{P}(Y(t) \g X)$ remains the same. 
 Under covariate shift,  WCP produces a valid predictive interval \citep{tib2019conformal}.

WCP has been applied to the ITE estimation \citep{lei2020conformal}.
Suppose we want to estimate the missing outcome $Y(t)$ of a randomly sampled unit.
The relationship between the observed data and the inference target $(X, Y(t))$ is %
\bac{
&\text{Training:}~~(X_i,Y_i(t)) \stackrel{i.i.d.}{\sim} p(X \g T=t)\cdot p(Y(t) \g X,T=t), \quad i \in \{i: T_i = t\}; \\
&\text{Target:}~~~(X, Y(t)) \sim p(X)\cdot p(Y(t) \g X). 
\label{eq:test-point}
}
For the training data, we observe both the covariates $X_i$ and outcome $Y_i(t)$. For a target data point, we only observe covariates $X$ and the goal is to infer the missing outcome $Y(t)$. 

WCP for the ITE estimation crucially relies on the assumption of \emph{unconfoundedness}  \citep{kivaranovic2020conformal,lei2020conformal},  that the units are assigned to the treatment groups based only on the observed covariates, i.e. $(Y_i(0), Y_i(1)) \indep T_i \g X_i$. Under unconfoundedness,   the conditional distributions of a potential outcome remain invariant across treatment groups, $P(Y(t) \g X, T=t)=P(Y(t) \g X)$. And \Cref{eq:test-point}  becomes
\bac{
&\text{Training:}~~(X_i,Y_i(t)) \stackrel{i.i.d.}{\sim} p(X \g T=t)\cdot p(Y(t) \g X), \quad i \in \{i: T_i = t\}; \\
&\text{Target:}~~~(X, Y(t)) \sim p(X)\cdot p(Y(t) \g X). 
\label{eq:train-test}
}
The only difference between the training and the target distribution is on the covariates distribution.   In other words,  unconfoundedness reduces the setting of counterfactual inference from the general distribution shift in \Cref{eq:test-point} to covariate shift in \Cref{eq:train-test}, for which  WCP is readily applicable.%

\subsection{ ITE Estimation under Unconfoundedness}\label{subsec:ITE}
\label{sec:ite-nuc}

We explain how WCP tackles the predictive inference problem in \Cref{eq:test-point} and point out the challenges when unconfoundedness is violated. 

Denote each training data pair as $Z_i \vcentcolon= (X_i, Y_i(t))$, $Z_{1:n} = (Z_1, \cdots, Z_n)$. Given a predictive function $\hat{\mu}(\cdot): \mathcal{X} \mapsto \mathcal{Y}$, conformal inference uses  a scalar-valued function $V: \mathcal{X}\times \mathcal{Y} \to \bR$ to measure the predictive error.  For instance, $V_i$ might be chosen as the absolute residual function $V(X_i,Y_i(t)) = |Y_i(t) - \hat{\mu}(X_i)| $ with mean prediction $\hat{\mu}(\cdot)$ \citep{vovk2005algorithmic}; we will  discuss how to fit the predictor $\hat{\mu}(\cdot)$ later and take it as a fixed mapping for now. The \emph{nonconformity score}  is  defined for each data point as $V_i \vcentcolon= V(Z_i)$.  %

Denote the  $\alpha$-th quantile of a random variable $X \sim p(X)$ as $Q_\alpha(X)$, where $Q_\alpha(X) = \inf \{x : p(X\leq x) \geq \alpha\}, ~\alpha \in [0,1]$. Let $\delta_v(V)$ be the Dirac delta function, defined as $\delta_v(V)=1$ if $V=v$ and $\delta_v(V)=0$ otherwise.  We denote the quantile of a discrete distribution and the quantile of an empirical distribution as
\bas{
Q_\alpha \left(\sum_{i=1}^n p_i\delta_{v_i}\right) \vcentcolon= Q_\alpha(V), ~ V \sim \sum_{i=1}^n p_i\delta_{v_i}; \quad Q_\alpha(v_{1:n}) \vcentcolon= Q_\alpha \left(\sum_{i=1}^n \frac1n\delta_{v_i}\right). 
}  
We define the \emph{conformal weights} as the density ratio of the training and target distributions,
\ba{
w_t(x,y) \vcentcolon= \frac{p(X=x)p(Y(t)=y \g X=x)}{p(X=x \g T=t)p(Y(t)=y \g X=x,T=t)}.
\label{eq:weight-raw}
}
Then for units $1\leq i \leq n$, let the normalized weights be 
\bas{
p_i^t(Z_{1:n}, (x,y)) \vcentcolon= \frac{w_t(Z_i)}{\sum_{i=1}^{n}w_t(Z_i) + w_t(x,y)}, ~ p^t_{n+1}(Z_{1:n},(x,y)) \vcentcolon= \frac{w_t(x,y)}{\sum_{i=1}^{n}w_t(Z_i) + w_t(x,y)}.%
}
where $p_i^t$ are the weights for the observed data and $p^t_{n+1}$ is the weight for a new data. When the conformal weights in \Cref{eq:weight-raw} are known or computable, we can use the WCP    to derive a predictive interval \citep{tib2019conformal}, 
\ba{
\hat{C}_t(x) = \big\{y \in \bR: V{(x,y)} \leq Q_{1-\alpha}\big(   \sum_{i=1}^n p_i(Z_{1:n},(x,y)) \delta_{V_i} + p_{n+1}(Z_{1:n},(x,y)) \delta_{\infty}    \big)\big\}.
\label{eq:infeasible-interval}
}
The interval $\hat{C}_t(x)$ is guaranteed with a pre-set $1-\alpha$ coverage probability \citep{lei2020conformal,tib2019conformal}, i.e. $\mathbb{P}_{(X,Y(t)) \sim p(X)p(Y(t) \g X)}(Y(t) \in \hat{C}_t(X)) \geq 1-\alpha$ .

Computing the predictive interval $\hat{C}_t(x)$ relies on the conformal weights being  accessible. In an ideal randomized controlled trial (RCT) with perfect compliance, the training and target data in \Cref{eq:test-point} are i.i.d., hence the conformal weights $w_t(x,y)  \equiv 1$. In an observational study under unconfoundedness,  the conformal weights  $w_t(x,y)  = p(X=x)/p(X=x | T=t)$ can be estimated from the observed data \citep{lei2020conformal}. 

 However, when unconfoundedness is violated, the joint distribution of the covariates and  outcome shifts from  training to  target as shown in  \Cref{eq:test-point}. We will see that the conformal weights $w_t(x,y)$ are nonidentifiable under such distribution shift. When  unconfoundedness is violated, existing conformal inference  cannot be directly applied to ITE estimation. ~\looseness=-1

\section{Sensitivity Analysis for ITEs}
\label{sec:csa}
In this section, we  develop an individual-level sensitivity analysis  for estimating a missing  outcome  and  generalize it to a sensitivity analysis for the ITE. 
We first define what it means to deviate from unconfoundedness.  We then show how to incorporate the uncertainty from an unknown confounding into the construction of a valid predictive interval.  ~\looseness=-1

\subsection{Confounding Strength and the Marginal Sensitivity Model}
\label{sec:model}

A sensitivity analysis quantifies the deviation from unconfoundeness and evaluates the corresponding range of causal estimates. 
We quantify the strength of unmeasured confounding  by the marginal sensitivity model (MSM) \citep{tan2006distributional,zhao2019sensitivity}.  
 Under  unconfoundeness,  the \emph{propensity score}, $e(x) := p(T=1  \g  X=x) $  \citep{rosenbaum1983central}, is the same as the  \emph{selection score}, $s_t(x,y) := p(T=1  \g  X=x, Y(t)=y)$ \citep{robins2000sensitivity,scharfstein1999adjusting}.
 Without unconfoundeness, the selection scores no longer equal to the propensity score. Their difference represents the strength of confounding, which can be measured by the odds ratio (OR), 
$
\text{OR}(s_t(x,y),e(x)) \vcentcolon= \left[s_t(x,y)/(1-s_t(x,y))\right] \left[e(x)/(1-e(x))\right].
$
The MSM assumes that under the true data distribution $\mathbb{P}_0$,  the odds ratio between the selection score and the propensity score is bounded by a given range \citep{tan2006distributional,zhao2019sensitivity}.

\begin{definition}[Marginal Sensitivity Model] Under the distribution $\mathbb{P}_0$ over $(X, T, Y(1), Y(0))$, assume the propensity score $e(x)$ and selection score $s_t(x,y)$ satisfy $s_t(x,y) \in \cE(\Gamma) $, where 
\bac{
\cE(\Gamma) = \{s(x,y): &s(x,y):  \mathcal{X} \times  \mathcal{Y} \mapsto [0,1], \\
&1/\Gamma  \leq \text{\normalfont OR}(s(x,y),e(x)) \leq \Gamma, \text{for all}~x\in \mathcal{X}, y\in \mathcal{Y} \},
\label{eq:t-model}
}
and  the sensitivity  parameter $\Gamma \geq 1$.
\end{definition}

 The magnitude of $\Gamma$ is the degree of deviation from unconfoundedness.   The set $\cE(\Gamma)$ expands with an increasing $\Gamma$, representing more possible ways of the treatment assignment that are not explained by the observed covariates. By specifying $\Gamma$, the MSM assumes a rich set of data generating distributions, which avoids imposing parametric assumptions on how an unobserved confounder interacts with the treatment and outcome.

\iffalse
Following \citet{zhao2019sensitivity}, denote $h(x,y) \vcentcolon  = \log(\text{OR}(s_t(x,y), e(x)))= \text{logit}(s_t(x,y)) - \text{logit}(e(x))$. The set $\cE(\Gamma) $ in \Cref{eq:t-model} can be equivalently written as 
\bac{
\cE(\Gamma) = \{s_t(x,y): s_t(x,y) = \sigma(h(x,y) + \text{logit}(e(x))),~h \in \mathcal{H}(\lambda)\}, 
\label{eq:h-model}
}
where $ \mathcal{H}(\lambda)=\{h: \mathcal{X} \times  \mathcal{Y} \to \bR: \norm{h}_{\infty} \leq \lambda\}$,  $\sigma(z)=1/(1+\exp(-z))$,  and $ \lambda = \log(\Gamma)$. Hereafter, we will refer to $h(x,y) \in \mathcal{H}(\lambda)$ as the sensitivity model. Note that $h(x,y)$ is defined independently of the propensity score.
\fi

Under the assumption of an MSM, we will develop  CSA in a two-stage approach. In the first stage, we quantify how the uncertainty from unmeasured confounding propagates to the conformal weights. In the second stage, we leverage the uncertainty in the conformal weights to create a valid predictive interval.

\subsection{From Confounding Strength to the Conformal Weights}

\label{sec:to-weights}

We  illustrate that under unmeasured confounding, the predictive interval in \Cref{eq:infeasible-interval} cannot be computed by the observational data, which is the main challenge in applying conformal inference for sensitivity analysis. 

The conformal weights  in \Cref{eq:weight-raw} decompose to two  terms.   The first term $p(X)/p(X \g T=t) = p(T=t)/p(T=t \g X)$ can be inferred from data. 
The second term  is,
\ba{
\frac{p(Y(t) \g X)}{p(Y(t) \g X,T=t)} = p(T=t \g X) + \frac{p(Y(t) \g X,T=1-t)}{p(Y(t) \g X,T=t)} p(T=1-t \g X). 
\label{eq:ratio}
}
Without  unconfoundedness, the density ratio on the right hand side of  \Cref{eq:ratio} involves the nonidentifiable distribution $p(Y(t) \g X,T=1-t)$ of the missing  potential outcome.

To deal with this challenge, we transfer the uncertainty from the unknown confounding to the uncertainty of the conformal weights. The nonidentifiable density ratio  term in the conformal weights is related to the odds ratio in  the MSM. Applying Bayes rule, 
\ba{
p(T=1 \g X,Y(t)) = & \frac{p(Y(t)  \g  X, T=1)p(T=1 \g X)}{p(Y(t) \g X)}  \notag \\
= &1/ \Big(1+ \frac{1-e(X)}{e( X)}\frac{p(Y(t) \g X,T=0)}{p(Y(t) \g X, T=1)}\Big).
\label{eq:sscore}
}
\Cref{eq:sscore} is also known as  Tukey's factorization \citep{brook1964distinction, franks2019flexible,franks2016non}.

Based on  Tukey's factorization, the following lemma shows that the conformal weight is proportional to the inverse selection score and the density ratio of a potential outcome in the two treatment groups is bounded by the sensitivity parameter.

\begin{lemma}
(i) For the conformal weights in \Cref{eq:weight-raw}, we have $w_t(x,y) = p(T=t)/p(T=t  \g  X=x, Y(t)=y)$.  (ii) The MSM with parameter $\Gamma$  equivalently  assumes
\ba{ 
\frac1\Gamma \leq \frac{p(Y(t)=y | X=x, T=1)}{p(Y(t)=y | X=x, T=0)} =  \text{\normalfont OR}(s(x,y),e(x)) \leq \Gamma 
\label{eq:y-model}
}
\label{lm:sscore} \vspace{-12mm}
\end{lemma}

The sensitivity parameter of the MSM specifies the range of plausible conformal weights.

\begin{lemma}
\label{thm:weight}
Given an MSM with sensitivity parameter $1 \leq \Gamma < \infty$, the weight function for the weighted conformal prediction in \Cref{eq:weight-raw} is bounded by 
\ba{
\Big(1+\frac{1}{\Gamma }\big(\frac{1-e(x)}{e(x)}\big)^{2t-1}\Big)p(T=t) \leq  w_t(x,y)  \leq \Big(1+\Gamma\big(\frac{(1-e(x))}{ e(x)}\big)^{2t-1}\Big)p(T=t) 
\label{eq:range}
}
\label{lm:range}
\vspace{-12mm}
\end{lemma}
\noindent Note that the bounds in \Cref{eq:range} are uniform for all $y$. When $\Gamma=1$, the upper and lower bounds of the conformal weights are the same. When $\Gamma>1$, the conformal weights can not be point identified.  The range in \Cref{eq:range} represents the weight uncertainty.

\subsection{From Conformal Weights to the Predictive Band}
\label{sec:to-band}

We first define a valid predictive band in sensitivity analysis. Then  we demonstrate the validity of a predictive interval given a specific sensitivity model. The union of such intervals becomes a valid predictive band for the MSM. Finally, as a practical way to obtain the union set, we propose and solve a constrained quantile optimization problem.

\noindent\textbf{Valid Predictive Bands under Sensitivity Models.~~}   %
By \Cref{eq:y-model}, each selection score $s_t(x,y)$ specifies a missing outcome distribution
\ba{
p^{(s_t)}(Y(t)  \g  X, T=1-t)=  \text{OR}(s_t(X,Y(t)), e(X))^{1-2t} \cdot p(Y(t) \g X, T=t).
\label{eq:counterfactual}
}
For a selection score $s_t$ in the collection of sensitivity models $\cE(\Gamma)$, the target data in \Cref{eq:test-point}  is generated by a distribution depending on $s_t$, i.e.,
\ba{
\resizebox{.92 \textwidth}{!} 
{
    $ p^{(s_t)}(Y(t) \g X) = p(T=t \g X) p(Y(t) \g X,T=t) + p(T=1-t \g X)p^{(s_t)}(Y(t)  \g  X, T=1-t).  $
}
\label{eq:outcome-s}
}

With the notation above, the validity of the predictive interval is defined as a worst case guarantee under \emph{all} plausible sensitivity models in  $\cE(\Gamma)$.

\begin{definition}
Under a set of sensitivity models $\cE(\Gamma)$, 
the predictive band for the potential outcome $Y(t)$ with  $(1-\alpha)$ marginal coverage is a band that satisfies
\begin{equation}
\mathbb{P}_{X,Y(t) \sim p(X)p^{(s_t)}(Y(t) \g X)}(Y(t) \in \hat{C}_t(X)) \geq 1-\alpha, \label{eq:interval}
\end{equation}
for any data generating distribution $\mathbb{P}_0$ with the corresponding selection score $s_t \in \cE(\Gamma)$.
\label{def:interval}
\end{definition}
The goal is  to construct  a predictive interval that satisfies \Cref{def:interval} with $\cE(\Gamma)$ defined by the MSM. Our first step is to create a  valid predictive interval under a sensitivity model.

\noindent\textbf{Coverage Guarantees for a Fixed Sensitivity Model.~~}   Given a fixed $s_t \in \cE(\Gamma)$, plugging \Cref{eq:outcome-s} to \Cref{eq:weight-raw},  the conformal weight $w_t^{(s_t)}(x,y)$ becomes a function of $s_t$. Let $w_i^{(s_t)} = w_t^{(s_t)}(Z_i) $ be the conformal weight for $Z_i$,
the predictive band  in \Cref{eq:infeasible-interval}  becomes
\ba{
\hat{C}^{(s_t)}(x) = \big\{y \in \bR: V{(x,y)} \leq Q_{1-\alpha}\big(   \sum_{i=1}^n p_i^{(s_t)} \delta_{V_i} + p_{n+1}^{(s_t)} \delta_{\infty}    \big)\big\},
\label{eq:h-interval}
}
where $\{p_i^{(s_t)}\}_{i=1}^{n+1}$ normalizes $\{w_i^{(s_t)}\}_{i=1}^{n+1}$ and $n$ is the number of data points.  %

In the following theorem, we show that $\hat{C}^{(s_t)}(x) $ has a valid coverage given  a specific sensitivity model $s_t$, when the propensity score is either known or estimable.

\begin{lemma}
Under SUTVA and strong overlap, for a selection score $s_t \in \cE(\Gamma)$, we have 
\begin{itemize} 
\item[{(1)}] With a known propensity score $e(X)$, the predictive band in \Cref{eq:h-interval} has  coverage
\ba{
1-\alpha \leq \mathbb{P}_{X,Y(t) \sim p(X)p^{(s_t)}(Y(t) \g X)}(Y(t) \in \hat{C}^{(s_t)}_t(X) ) \leq 1-\alpha +  \frac{\Gamma/\eta}{n + \Gamma/\eta} 
\label{eq:true-e}
}
\item[{(2)}] With an estimated propensity score $\hat{e}(X)$,  if $\eta <\hat{e}(X_i)<1-\eta$ almost surely for a constant $\eta \in(0,0.5)$, the predictive band $\hat{C}^{(s_t)}(x) $ in \Cref{eq:h-interval} has a coverage probability 
\ba{ \label{eq:hat-e}
&1-\alpha - \Delta \leq \mathbb{P}(Y(t) \in \hat{C}^{(s_t)}(X)) \leq 1-\alpha +  \frac{\Gamma/\eta}{n + \Gamma/\eta} + \Delta, \\
&\Delta = \frac{\Gamma}{2} p(T=t)\E_{x\sim p(X \g T=t)} \big| \frac{1}{\hat{e}(x)^t(1-\hat{e}(x))^{1-t}} - \frac{1}{e(x)^t(1-e(x))^{1-t}} \big|. \notag
}

\end{itemize} 
\label{thm:condition}
\end{lemma}

With a known propensity score, \Cref{eq:true-e} demonstrates that the coverage of the predictive band $\hat{C}^{(s_t)}(x) $ is valid and is close to the nominal level.  The closeness depends on the overlapping and confounding strength.  When the estimated propensity score $\hat{e}(X) $ differs from the true propensity score $e(X)$, as shown in  \Cref{eq:hat-e}, the coverage probability might have an extra slack quantity $\Delta$. The reason for $\hat{e}(X) \neq e(X)$ could be the estimation error from the finite sample or  the inference error from a misspecified treatment model. %

\noindent\textbf{Union Method.~~}  Based on the predictive band $\hat{C}^{(s_t)}(x)$, we propose a union method that achieves the valid coverage under the MSM. That is, we now consider the worst case coverage for all sensitivity models $s_t \in \cE(\Gamma)$. 

\begin{proposition}
Suppose the predictive interval $\hat{C}^{(s_t)}(X) = [L^{(s_t)}(X),U^{(s_t)}(X)]$ satisfies 
\ba{
\mathbb{P}_{X,Y(t) \sim p(X)p^{(s_t)}(Y(t) \g X)}(Y(t) \in \hat{C}^{(s_t)}_t(X) ) \geq 1-\alpha \label{eq:condition}
}
for each $s_t \in \mathcal{E}(\Gamma)$. Then let $L = \inf_{s_t \in \mathcal{E}(\Gamma)} L^{(s_t)}$ and $U = \sup_{s_t \in \mathcal{E}(\Gamma)} U^{(s_t)}$, the interval $ \hat{C}_t^\Gamma(X) = [L, U] = \cup_{s_t \in \mathcal{E}(\Gamma)}[L^{(s_t)},U^{(s_t)}]$ is a predictive interval for $Y(t)$ with at least $(1-\alpha)$ coverage under the sensitivity models $ \mathcal{E}(\Gamma)$.
\label{prop:union}
\end{proposition}

\Cref{prop:union} states that to obtain a valid predictive interval for the MSM, we can first compute the predictive interval under  a specific sensitivity model as in \Cref{eq:h-interval}, then take the union set by finding the extreme end points of such intervals over all the sensitivity models. By \Cref{eq:h-interval}, finding the extreme end points is equivalent to solving
\ba{
\max_{s_t \in \mathcal{E}(\Gamma)}  ~Q_{1-\alpha}\Big(   \sum_{i=1}^n p_i^{(s_t)}\big(Z_{1:n}, (X,y)\big)\delta_{V_i} + p_{n+1}^{(s_t)}\big(Z_{1:n}, (X,y)\big) \delta_{\infty}  \Big).
\label{eq:optimize-h}
}
However, in practice, it is difficult to directly search over the sensitivity models in $\mathcal{E}(\Gamma)$, because the elements of  $\mathcal{E}(\Gamma)$ are not defined parametrically.

\noindent\textbf{Quantile Optimization with Linear Constraints.~~}  To operationalize \Cref{eq:optimize-h},
we can search over the conformal weights instead of the sensitivity models. As \Cref{eq:optimize-h} shows,  a sensitivity model influences the predictive band only through the conformal weights. 

In \Cref{sec:to-weights}, we find the range of conformal weights under an MSM. Denote the upper and lower bounds of the conformal weights in \Cref{lm:range} as
\ba{
\resizebox{0.92\hsize}{!}{%
$w_{lo}^{\Gamma}(x) \vcentcolon = \Big(1+\frac{1}{\Gamma }\big(\frac{1-e(x)}{e(x)}\big)^{2t-1}\Big)p(T=t), ~~w_{hi}^{\Gamma}(x) \vcentcolon=\Big(1+\Gamma\big(\frac{(1-e(x))}{ e(x)}\big)^{2t-1}\Big)p(T=t).$
}
\label{eq:weightbound}
}
 Then the optimization in \Cref{eq:optimize-h} simplifies to a constrained optimization problem,
\bac{
\max_{w_{1:n+1}}  \quad &Q_{1-\alpha}\big(   \sum_{i=1}^n p_i \delta_{V_i} + p_{n+1}\delta_{\infty}  \big). \\
  \text{subject to} \quad& p_i = \frac{w_i}{\sum_{i=1}^{n+1} w_i}, ~~ 1 \leq i \leq n+1 \\
 w_{lo}^{\Gamma}(X_i) \leq w_i &\leq w_{hi}^{\Gamma}(X_i), ~1 \leq i \leq n, ~~w_{lo}^{\Gamma}(X) \leq w_{n+1} \leq w_{hi}^{\Gamma}(X),
\label{eq:qt-opt}
}
where the conformal weights $\wv = (w_1, \cdots, w_n, w_{n+1})$  are the optimizing variables. For notational convenience, we suppress the superscript $t$ and denote $w_i = w(X_i, Y_i)$ for $1\leq i \leq n$,  and $w_{n+1} = w(X, y)$. %

An efficient algorithm to solve \Cref{eq:qt-opt} can be designed  by characterizing its optima.

{\LinesNumberedHidden
\begin{algorithm}[!t]
  \caption{CSA for Estimating an Unobserved Potential Outcome} %
  \label{alg:csa}
\textbf{Input: } Data $\Z=(X_{i}, Y_{i}, T_{i})_{i=1}^{N}$, where $Y_{i}$ is missing if $T_{i}=1-t$; level $\alpha$, sensitivity parameter $\Gamma$, target point covariates $X$. \\
  \vspace{0.2cm}
  \textbf{Step I: Preliminary processing}
  \begin{algorithmic}[1]
    \State Split the data into two folds $\Z_{\text{pre}}$ and $\Z_{\text{cal}}$; let  $\I_{\text{pre}} = \{i: Z_i \in \Z_{\text{pre}}, T_i = t\}$,  $\I_{\text{cal}} = \{i: Z_i \in \Z_{\text{cal}}, T_i = t\}$ %
    \State Estimate propensity score $\hat{e}(x)$ on $\Z_{\text{pre}}$
    \State Estimate predictor $\hat{\mu}(\cdot)$ on $\{X_i, Y_i\}_{i\in \I_{\text{pre}}}$
     \end{algorithmic}

  \vspace{0.2cm}
  \textbf{Step II: Predictive interval for $Y_i(t)$ at the target point}
   \begin{algorithmic}[1]
  \State Compute nonconformity scores $\mathcal{V} = \{V_i\}_{i \in \tilde{\I}_{\text{cal}} }$, $\tilde{\I}_{\text{cal}} = \I_{\text{cal}} \cup \{N+1\}$, $V_{N+1}=\infty$ 
 
  \State Compute the bounds $(w_{lo}^{\Gamma}(X_i), w_{hi}^{\Gamma}(X_i))$ for $i \in \tilde{\I}_{\text{cal}} $ by \Cref{eq:weightbound}, $X_{N+1}=X$;   \State For $i \in \tilde{\I}_{\text{cal}} $, initialize the weights $w_i = w_{lo}^{\Gamma}(X_i)$ 
  \State Sort $\mathcal{V}$ in ascending order and re-label the ordered elements  from 1 to $|\mathcal{V}|$
 \State Re-label $\{w_i\}_{i \in \tilde{\I}_{\text{cal}}}$, $\{X_i\}_{i \in \tilde{\I}_{\text{cal}}}$ according to the labels of the sorted $\mathcal{V}$; set $k = |\mathcal{V}|$
\Do
     \NoNumber{$w_{k} \leftarrow w_{hi}^{\Gamma}(X_{k})$
    \State Compute normalized weights $p_i = \frac{w_i}{\sum_{j=1}^{|\mathcal{V}|}w_j}$, for $i \in \idx{|\mathcal{V}|}$
    \State$k \leftarrow k-1$
  \doWhile{$\sum_{i = k+1}^{|\mathcal{V}|} p_i < \alpha$}
  }
   \end{algorithmic}
  \textbf{Output:} Compute  $\hat{C}_t^{\Gamma}(X)$ by \Cref{eq:csa-interval} with $\hat{Q}(Z_{1:n},X) = V_{k+1}$  %
\end{algorithm}
}

\begin{proposition}
Without loss of generality, suppose $X_{1:n}$ are labeled such that the nonconformity scores are ordered, $V_1 \leq V_2 \leq \cdots \leq V_n < V_{n+1}= \infty$,
and  let %
\bas{
\hat{k} = \max\big\{k \in [n+1] : &\text{for}~k \leq j \leq n,~ w_j = w_{hi}^{\Gamma}(X_j), w_{n+1} =  w_{hi}(X), ~ \sum_{j=k}^{n+1}p_j \geq \alpha;  \\
&\text{for}~ j < k,  w_j = w_{lo}^{\Gamma}(X_j) \big\}
}
Then the optima of \Cref{eq:qt-opt} is $\hat{\wv} = (w_{lo}^{\Gamma}(X_1), \cdots, w_{lo}^{\Gamma}(X_{\hat{k}-1}), w_{hi}^{\Gamma}(X_{\hat{k}}), \cdots, w_{hi}^{\Gamma}(X_{n}), w_{hi}^{\Gamma}(X))$. Furthermore, the optimal objective value $ \hat{Q}(Z_{1:n}, X; \Gamma, \alpha, t) = V_{\hat{k}}$. 
\label{prop:qt-opt}
\end{proposition}

According to \Cref{prop:qt-opt}, to solve \Cref{eq:qt-opt}, we first sort $V_{1:n+1}$ in ascending order and initialize the conformal weights $w_i = w_{lo}^{\Gamma}(X_i)$ for $1\leq i \leq n$, $w_{n+1} = w_{lo}^{\Gamma}(X)$. Then we iteratively flip $w_k$ from $w_{lo}^{\Gamma}(X_k)$ to $w_{hi}^{\Gamma}(X_k)$ for $k=n+1, n, \cdots, 1$ until $\sum_{i=k}^{n+1} p_i \geq \alpha$. Supposing the iteration stops at $k=m$, the optimal  objective value in  \Cref{eq:qt-opt} is uniquely determined as $ \hat{Q}(Z_{1:n}, X; \Gamma, \alpha, t) =V_{m}$.   

To sum up, by \Cref{thm:condition}, the interval $\hat{C}_t^{(s_t)}(x) $ in \Cref{eq:h-interval} satisfies the coverage of \Cref{eq:true-e,eq:hat-e}. With $\hat{Q}$ given in \Cref{prop:qt-opt}, 
\ba{
\hat{C}_t^{\Gamma}(x) = \big\{y \in \bR: V{(x,y)} \leq \hat{Q}(Z_{1:n}, x; \Gamma, \alpha, t)  \big\}  \label{eq:csa-interval}
}
is the union set $ \cup_{s_t \in \cE(\Gamma)} \hat{C}^{(s_t )}(x)$. %
According to  \Cref{prop:union},  $\hat{C}_t^{\Gamma}(x)$ is a valid  predictive interval of a missing counterfactual outcome under the MSM.
\begin{theorem}
Under the condition of \Cref{thm:condition}, with known propensity score $e(X)$, the predictive band $\hat{C}_t^{\Gamma}(x)$ in \Cref{eq:csa-interval} has nominal coverage $1-\alpha$ of $Y(t)$ under the collection of sensitivity models $\cE(\Gamma)$; with an estimated $\hat{e}(X)$, the coverage is at least to the lower bound in \Cref{eq:hat-e}.  
\end{theorem}

\noindent\textbf{Implementation and Computational Cost.~~}  We compute the predictive interval by adopting the framework of split conformal prediction \citep{lei2014distribution,papadopoulos2002inductive}.  Classic conformal inference fits the predictive function using the leave-one-out observed data to ensure the exchangeability %
and has to fit a predictor multiple times. Spilt conformal prediction  reduces the computational cost by  randomly splitting the observed data into a preliminary set and a calibration set. The prediction model is fitted on the preliminary set for one time, set as fixed, and used to compute the nonconformity scores on the calibration set and target set.

For CSA, the predictive interval in \Cref{eq:csa-interval} can be computed analytically on top of a specific conformal inference algorithm.  %
As an example,  for the split conformal inference with nonconformity score $V_i = |Y_i - \hat{\mu}(X_i)|$ \citep{lei2015conformal}, where $\hat{\mu}(\cdot)$ is the mean response function, 
 \Cref{eq:csa-interval} becomes
$
\hat{C}_t^{\Gamma}(x)  = [\hat{\mu}(x) - \hat{Q}(Z_{1:n}, x;\Gamma, \alpha),~ \hat{\mu}(x) + \hat{Q}(Z_{1:n}, x; \Gamma, \alpha)]. %
$
For split conformal quantile regression with nonconformity score $V_i = \max \{\hat{q}_{\alpha/2}(X_i)-Y_i, Y_i-\hat{q}_{1-\alpha/2}(X_i)\}$ \citep{lei2020conformal,romano2019conformalized}, where $\hat{q}(\cdot)$ is the conditional quantile function, \Cref{eq:csa-interval} becomes
$
\hat{C}_t^{\Gamma}(x)  = [\hat{q}_{\alpha/2}(x) - \hat{Q}(Z_{1:n}, x;\Gamma, \alpha), ~\hat{q}_{1-\alpha/2}(x) + \hat{Q}(Z_{1:n}, x;\Gamma,\alpha)]. %
$
The full algorithm  is summarized in \Cref{alg:csa}.

For each target unit, to solve the optimization in \Cref{eq:qt-opt}, the computational complexity is $\mathcal{O}(mn)$ if the loop ends in $m$ iterations and the worst-case complexity is $\mathcal{O}(n^2)$. When the target coverage $1-\alpha$ is high, $m$ is close to 1 and the total computation time is close to  the optimal rate that is needed to evaluate the objective function for one time. Other computations are sorting $V_{1:n+1}$ and fitting the treatment and outcome models, which can be shared by different target units. Therefore,  CSA is highly efficient, inducing little extra computation comparing to the conformal prediction under unconfoundedness.

\subsection{Predictive Band for the Individual Treatment Effect}
\label{sec:ite}

We now develop a sensitivity analysis for the ITE of a target unit, for which both potential outcomes are  unobserved. 
Let the covariates of a target unit be $X $. Using  the data of the treatment group $t$, by \Cref{alg:csa}, we can construct an interval $\hat{C}_t^{\Gamma}(X )= [L_t^{\Gamma}(X ), U_t^{\Gamma}(X )]$ which  has $1-\alpha_t$ coverage of $Y (t)$ under the sensitivity models $\cE(\Gamma)$. Let $\hat{C}^{\Gamma}(X) = [L_1^{\Gamma}(X )-U_0^{\Gamma}(X ), U_1^{\Gamma}(X )-L_0^{\Gamma}(X)]$ and  $\alpha_1 + \alpha_0 = \alpha$. By the Bonferroni correction, 
\ba{
\mathbb{P}(Y(1)-Y(0) \in \hat{C}^{\Gamma}(X )) \geq 1- \mathbb{P}(Y(1) \notin \hat{C}_1^{\Gamma}(X )~\text{or}~Y(0) \notin \hat{C}_0^{\Gamma}(X )) \geq 1-\alpha.
\label{eq:ite-cover}
} 
So the predictive interval $\hat{C}^{\Gamma}(X )$ has the desired coverage.  Though computationally simple, the Bonferroni method might produce overly conservative interval $\hat{C}^{\Gamma}(X )$ for the ITE, because  the coverage $1-\alpha_t$ for each potential outcome is higher than $1-\alpha$.

To mitigate this problem, we follow \cite{lei2020conformal} and develop a nested approach.  The idea is to first randomly sample a subset of data  as the validation set  (indexed by $ \mathcal{I}_{\text{val}}$) and set  the rest of the observed data as the non-validation set. For each individual $i \in \mathcal{I}_{\text{val}}$, %
let $\hat{C}^{\Gamma}(X_i ) = \hat{C}_1^{\Gamma}(X_i)-Y_i(0)$ if $T_i=0$ and $\hat{C}^{\Gamma}(X_i ) = Y_i(1) - \hat{C}_0^{\Gamma}(X_i )$ if $T_i=1$. The coverage probability decomposes as
\bas{
\mathbb{P}(Y_i(1)-Y_i(0) \in \hat{C}^{\Gamma}(X_i )) = \mathbb{P}(T_i=1)&\mathbb{P}(Y_i(0) \in \hat{C}_0^{\Gamma}(X_i )|T_i=1) \\
&+ \mathbb{P}(T_i=0)\mathbb{P}(Y_i(1) \in \hat{C}_1^{\Gamma}(X_i )|T_i=0).
}
If the  interval $\hat{C}_t^{\Gamma}(X_i )$ has a coverage probability of $Y_i(t)$ higher than $1-\alpha$, the coverage probability of  $\hat{C}^{\Gamma}(X)$ for the ITE is also higher than $1-\alpha$. 
The dataset $\tilde{\mathcal{D}} = \{X_i, \hat{C}^{\Gamma}(X_i)\}_{i \in \mathcal{I}_{\text{val}}}$ with $X_i  \stackrel{i.i.d.}{\sim} p(X)$ can be used to fit a predictive function $X \mapsto \hat{C}^{\Gamma}(X)$, which maps to a 
predictive interval $ \hat{C}^{\Gamma}(X)$ for a  data point with covariates $X \sim p(X)$. The mapping can be two regressions with the input as $X_i$, $i \in \mathcal{I}_{\text{val}}$ and the output as the upper and lower end points of $ \hat{C}^{\Gamma}(X_i)$ respectively. This becomes a relatively easy in-sample prediction problem.

\begin{algorithm}[t]
  \caption{ CSA for the ITE estimation with Nested Method}  
  \label{alg:ite}
\textbf{Input: } Data $(X_{i}, T_{i}, Y_{i}(T_i))_{i=1}^{N}$, level $\alpha$, sensitivity parameter $\Gamma$, target covariates X \\
  \vspace{0.2cm}
  \textbf{Step I: Preliminary processing}
  \begin{algorithmic}[1]
    \State Split the data into two folds, indexed by $\I$ and $\I_{\text{val}}$
    \State Denote the treated and control group data in $\I$ ($\I_{\text{val}}$) as $\I^t$, $\I^c$ ($\I_{\text{val}}^t$, $\I_{\text{val}}^c$) respectively
     \end{algorithmic}

  \vspace{0.2cm}
  \textbf{Step II: Predictive interval for the ITE $\tau_i$ at the target point} 
     \begin{algorithmic}[1]
  \State Run \Cref{alg:csa} with data in $\I$, $\I_{\text{pre}} \cup \I_{\text{cal}} = \I^t$ and for each target point $i \in \I_{\text{val}}^c$; the bounds of weight $w_{lo}^{\Gamma}(x) =  \hat{e}(x)/(\Gamma(1-\hat{e}(x)))$ and $w_{hi}^{\Gamma}(x) = \Gamma \hat{e}(x)/(1-\hat{e}(x))$; return $\hat{C}_1^{\Gamma}(X_i)$
  \State For $i \in \I_{\text{val}}^c$, compute $\hat{C}^{\Gamma}(X_i) = \hat{C}_1^{\Gamma}(X_i)-Y_i(0)$
  \State Run \Cref{alg:csa} with data in $\I$, $\I_{\text{pre}} \cup \I_{\text{cal}} = \I^c$ and for each target point $i \in \I_{\text{val}}^t$; the bounds of weight $w_{lo}^{\Gamma}(x) = (1-\hat{e}(x))/(\Gamma \hat{e}(x))$ and   $w_{hi}^{\Gamma}(x) = \Gamma(1-\hat{e}(x))/ \hat{e}(x)$; return $\hat{C}_0^{\Gamma}(X_i) $
    \State For $i \in \I_{\text{val}}^t$, compute $\hat{C}^{\Gamma}(X_i) = Y_i(1) - \hat{C}_0^{\Gamma}(X_i)$
   \State Learn the predictive function $X \to \hat{C}^{\Gamma}(X)$  with training data $\{X_i, \hat{C}^{\Gamma}(X_i)\}_{i \in \mathcal{I}_{\text{val}}}$; predict $\hat{C}^{\Gamma}(X)$ for the target data with the learned predictive function
   \end{algorithmic}
  \textbf{Output: $\hat{C}^{\Gamma}(X)$} %
\end{algorithm}

We use \Cref{alg:csa} to obtain the predictive intervals $\hat{C}_t^{\Gamma}(X_i)$ for the data points in the validation set.  The training data are from the treatment group $1-t$ in the non-validation set and the target data are from the treatment group $t$ in the validation set, $t=0,1$. Similar to  \Cref{thm:weight}, the bounds of the conformal weights $w_t(x,y)$ can be computed as 
\ba{
\frac{p(X \g T=t)p(Y(t) \g X,T=t)}{p(X \g T=1-t)p(Y(t) \g X,T=1-t)}  \in \big[\frac1\Gamma  \Big(\frac{e(x)}{1-e(x)}\Big)^{2t-1}, ~\Gamma \Big(\frac{e(x)}{1-e(x)}\Big)^{2t-1}\big]. 
}
The algorithm is  summarized in  \Cref{alg:ite}.

\subsection{Sharpness via Covariates Balancing}
\label{sec:sharpness}

\noindent\textbf{Notion of sharpness.~~} Consider the sharpness on the sensitivity models. A sharp sensitivity model should be data compatible and not have observational implications \citep{Dorn2021-vs,franks2019flexible}. For the MSM, we define the sharp MSM as
\ba{
\cE^*(\Gamma) = \{ s_t(x,y) \in \cE(\Gamma): \int p^{(s_t)}(Y(1)=y \g X=x, T=0) dy = 1\},
\label{eq:sharp-msm}
}
where $\cE(\Gamma)$ is defined in \Cref{eq:t-model} and $p^{(s_t)}(Y(1)=y \g X, T=0)$ is the induced counterfactual distribution  in \Cref{eq:counterfactual}. The sharp MSM is a subset of the selection scores  in $\cE(\Gamma)$ that induces proper counterfactual density.  By \Cref{lm:sscore}, for example, $\cE^*(\Gamma)$ excludes the  selection scores with an odds ratio $ \text{OR}(s_t(X,Y(t)), e(X))$ uniformly greater (or less than) one. %

Recent work improves the sharpness of the MSM in estimating the ATE \citep{Dorn2021-vs,zhao2019sensitivity}.  \citet{Dorn2021-vs} shows that the selection score $s_1$ is  data compatible if it satisfies the constraint $\E[\frac{T}{s_1(X,Y(1))} \g X] = 1$ (the unobserved confounder in \citet{Dorn2021-vs} is replaced with $Y(1)$). This constraint is equivalent to the constraint in \Cref{eq:sharp-msm} as shown in \Cref{prop:sharp}, \Cref{sec:cssa}.  We consider estimating $Y(1)$ for simplicity and the discussion applies to $Y(0)$ similarly. The derivation and computation details of this section are presented in \Cref{sec:cssa}.

\noindent\textbf{Sharpness by Covariates Balancing.~~}  The integral constraint in \Cref{eq:sharp-msm} is easy to interpret but is infeasible to compute because we often only observe one outcome value for a given  $X$. However, the constraint, equivalent to  $\E[\frac{T}{s_1(X,Y(1))} \g X] = 1$ by \Cref{prop:sharp},  indicates that for an arbitrary vector-valued function $g(X)$, 
\ba{
\E[\frac{g(X_i)T_i}{s_1(X_i,Y_i(1))} ] =  \E_{X_i}\big[g(X_i) \E[\frac{T_i}{s_1(X_i,Y_i(1))} \g X_i ] \big]= \E[g(X_i)]. 
\label{eq:relaxed}
}
By enforcing the condition in \Cref{eq:relaxed} with different covariates function $g(X)$, we can reduce $\cE(\Gamma)$ close to $\cE^*(\Gamma)$.  Since \Cref{eq:relaxed} holds similarly for the control group,  it  represents the  covariate balancing between the treated and control group. This means encouraging the covariate balancing  improves the sharpness of the MSM.

We incorporate the balancing condition \Cref{eq:relaxed} to CSA. By  \Cref{lm:sscore}, we transform the condition in \Cref{eq:relaxed} to  the constraints in the quantile optimization in  \Cref{eq:qt-opt}. Specifically, we optimize \Cref{eq:qt-opt} with additional constraints
\ba{
\frac{1}{N_t}\sum_{i: T_i=t} g_k(X_i)&w_i^{\Gamma} = \frac{1}{N} \sum_{i=1}^N \frac{T_i^t(1-T_i)^{1-t}}{\hat{e}(X_i)^t(1-\hat{e}(X_i))^{1-t}}g_k(X_i),  ~~ 1 \leq k \leq K,
\label{eq:qt-opt-sharp}
}
where $g_k(X)$, $k\in \idx{K}$ are the balancing functions specified by the researcher.  We call this algorithm conformalized sharp sensitivity analysis (CSSA) and summarize it in \Cref{alg:cssa}, \Cref{sec:cssa}. The optima of \Cref{eq:qt-opt} with additional constraints \Cref{eq:qt-opt-sharp} is smaller than that of \Cref{eq:qt-opt}, thereby reducing the size of predictive band  for the ITE estimates.  In simulations, we find choosing $g(X) = \hat{e}(X)$ effectively improves the sharpness.  Including additional balancing functions such as the quantile function of the outcome distribution \citep{Dorn2021-vs}, the identity function, and the derivative of the estimated propensity score \citep{imai2014covariate} may further reduce $\cE(\Gamma)$ to $\cE^*(\Gamma)$.

Finally, though the sharpness is necessary for claiming a causal estimate to be sensitive to unmeasured confounding, sensitivity analysis is often applied to corroborate a nonzero causal effect identified in the primary analysis. From this perspective, the deviation from the sharpness might be  interpreted as  conservativeness \citep{cinelli2020making,ding2016sensitivity,veitch2020sense}. For the ITE estimation, such conservativeness increases our confidence that a positive (or negative) ITE is indeed robust to unmeasured confounding when the sensitivity analysis suggests so.

\section{Practical Considerations of the Algorithms}
\label{sec:practical}
We now discuss  how to interpret the coverage probability of CSA, choose the conformal inference algorithms, calibrate the sensitivity parameter, and evaluate the ITE estimation.

\noindent\textbf{Marginal and Conditional Coverage.~~} The probability in the coverage statement of CSA is over both the covariates and the outcomes. Hence, the coverage guarantee should be interpreted in a marginal way instead of a conditional way. In other words, suppose the estimand $\tau$ is either the missing potential outcome $Y(t)$ or the ITE $Y(1)-Y(0)$. $\mathbb{P}(\tau \in \hat{C}(X ))$ means that if we construct a predictive band $\hat{C}(X)$ for a  unit randomly sampled as the target, the probability that $\hat{C}(X)$ captures $\tau$ is at least $1-\alpha$.  The randomness is over both the covariates and the potential outcomes. 

The marginal coverage measures the quality of prediction \emph{averaged} over the target units. It does not guarantee the coverage of $\tau$  for a given \emph{fixed} target unit. The loss of  conditional coverage is unavoidable if  no distributional assumption is imposed on the observed data \citep{barber2019limits}. However, it is possible for an algorithm with marginal coverage to achieve conditional coverage asymptotically,  under additional regularization conditions,  or to satisfy a relaxed conditional coverage definition \citep{barber2019limits,lei2020conformal}. We refer to \citet{barber2019limits} for a detailed discussion on the definitions, limitations,  and connections of different types of coverages. 

\noindent\textbf{Choice of Conformal Inference Methods.~~} CSA is compatible with a variety of conformal inference algorithms. The main components of a conformal inference are the prediction model and the corresponding nonconformity score \citep{Angelopoulos2021-ur}. %
The choice of  conformal inference algorithm hinges on the properties of the underlying outcome distribution, such as the homogeneity, skewness and multimodality.
A good choice of  inference method leads to high interpretability and small predictive interval. For example, to estimate the ITE, the predictive band as a single interval may be more interpretable than a nonconvex set \citep{Sesia2021-ei}; an algorithm  capturing individual heterogeneity might produce a shorter interval and  more informative   estimate \citep{romano2019conformalized}. Nevertheless, the coverage validity of the predictive band does not depend  on the  choice of conformal inference method. %

\noindent\textbf{Calibration of the Sensitivity Parameter.~~} In the MSM, the sensitivity parameter $\Gamma$ quantifies the confounding strength. While setting $\Gamma$ to a proper value requires domain knowledge, the observed data can provide useful reference \citep{hsu2013calibrating,imbens2003sensitivity,kallus2019interval}.  
In the definition of MSM, $\Gamma$ measures the effect of knowing a potential outcome on the treatment assignment. We can view the potential outcome as a type of covariate \citep{robins2000sensitivity} and compute the effect of an observed covariate on the  treatment assignment. Specifically, we compute  $\Gamma_{ij} = \text{OR}(e(X_i), e((X_{\backslash j})_i))$ as the effect of the j-th covariates on the treatment assignment of the i-th unit, where $e(X_{\backslash j})$ is the propensity score estimated without the j-th covariates. The domain experts  can assess a plausible magnitude of $\Gamma$ by referring to the magnitude of $\{\Gamma_{ij}\}_{i,j}$. Here, approximating $e(X_i)$ by $e((X_{\backslash j})_i))$ may introduce conservativeness to the estimated confounding strength. We refer to  \cite{cinelli2020making,veitch2020sense} for a  discussion on this issue.

\noindent\textbf{Evaluating the Predictive Band of an ITE. ~~} When evaluating the ITE estimation by simulations, we need to sample the true ITEs, which requires generating random samples of all the potential outcomes. To generate $Y_i(t) \sim p(Y(t)  \g  X_i)$, we can sample $T_i \sim \text{Bern}(e(X_i))$ and $Y_i(t) \sim p(Y(t)  \g  X_i, T_i)$. However, when $T_i = 1-t$, generating  $Y_i(t)$ depends on a sensitivity model, which is not defined parametrically in the MSM. To solve this problem, we propose a  rejection sampling method to generate counterfactual samples. The details of this sampling method is presented in \Cref{sec:rejection-sampling}.

\section{ Empirical Studies}
\label{sec:experiment}

In this section,  we  answer the following questions using synthetic data: can CSA provide a desired coverage?
Are the  predictive intervals of the ITE overly conservative? 
How do the predictive intervals of ITE compare to the interval estimates of population-level causal estimands? Finally, we illustrate how to apply CSA in an observational study. Implementation and examples in this paper are available at \url{https://github.com/mingzhang-yin/Conformal-sensitivity-analysis}.

\subsection{CSA for Estimating Counterfactual Outcome}\label{sec:exp-missing}

Following the  synthetic data generation in \cite{lei2020conformal,wager2018estimation}, the potential outcome $Y_i(1)$ is from
\bac{
&Y_i(1) = \E[Y_i(1) \g T_i=1,X_i] +  \epsilon_i, ~\epsilon_i \sim \cN(0,\sigma^2); \\
&\E[Y_i(1) \g T_i=1,X_i] = f(X_{i1})f(X_{i2}),~ f(x) = \frac{2}{1+\exp(-5(x-0.5))},
\label{eq:exp1y1}
}
where the covariates $X_i = (X_{i1}, \cdots,X_{id})^\top$, $X_{ij} \sim \text{Unif}(0,1)$. The propensity score is 
$e(X_i) = \frac14(1+\beta_{2,4}(1-X_{i1}))$,
where $\beta_{2,4}$ is the CDF of beta distribution with parameters $(2,4)$. We generate $n=3000$ training data points $(X_i, Y_i(1))_{i:T_i=1}$ with dimension of covariates $d = 20$. We take $75\%$ of the training data as the preliminary set and the rest as the calibration set. For a calibration set with $n_{\text{cal}}$ data points, the coverage probability of a new target data follows distribution $\text{Beta}(n_{\text{cal}}+1-\lfloor(n_{\text{cal}}+1)\alpha\rfloor, \lfloor(n_{\text{cal}}+1)\alpha\rfloor)$  \citep{Angelopoulos2021}. 
We set the nominal level $1-\alpha=0.8$ in this simulation. We consider two settings: in the homoscedastic case,  $\sigma \equiv 1$, and in the heteroscedastic case, $\sigma \sim \text{Unif}(0.5,1.5)$. 

For CSA, we use the split conformal prediction with mean prediction \citep{lei2015conformal,papadopoulos2002inductive} and  conformal quantile regression \citep{romano2019conformalized}, denoted as CSA-M and CSA-Q respectively.  We implement CSSA in \Cref{sec:sharpness} with mean prediction, denoted as CSSA-M and set the balancing constraints $g(X_i)$ in \Cref{eq:qt-opt-sharp}  as the estimated propensity score $\hat{e}(X_i)$. We report the ITE estimated under no unobserved confounding (NUC)  as a benchmark, denoted as ITE-NUC \citep{lei2020conformal}. For all conformal inference methods, we use the random forest \citep{breiman2001random} as the regression function.

\begin{figure}[t]
\centering{
  \subfloat
   {{\includegraphics[width=0.32\textwidth]{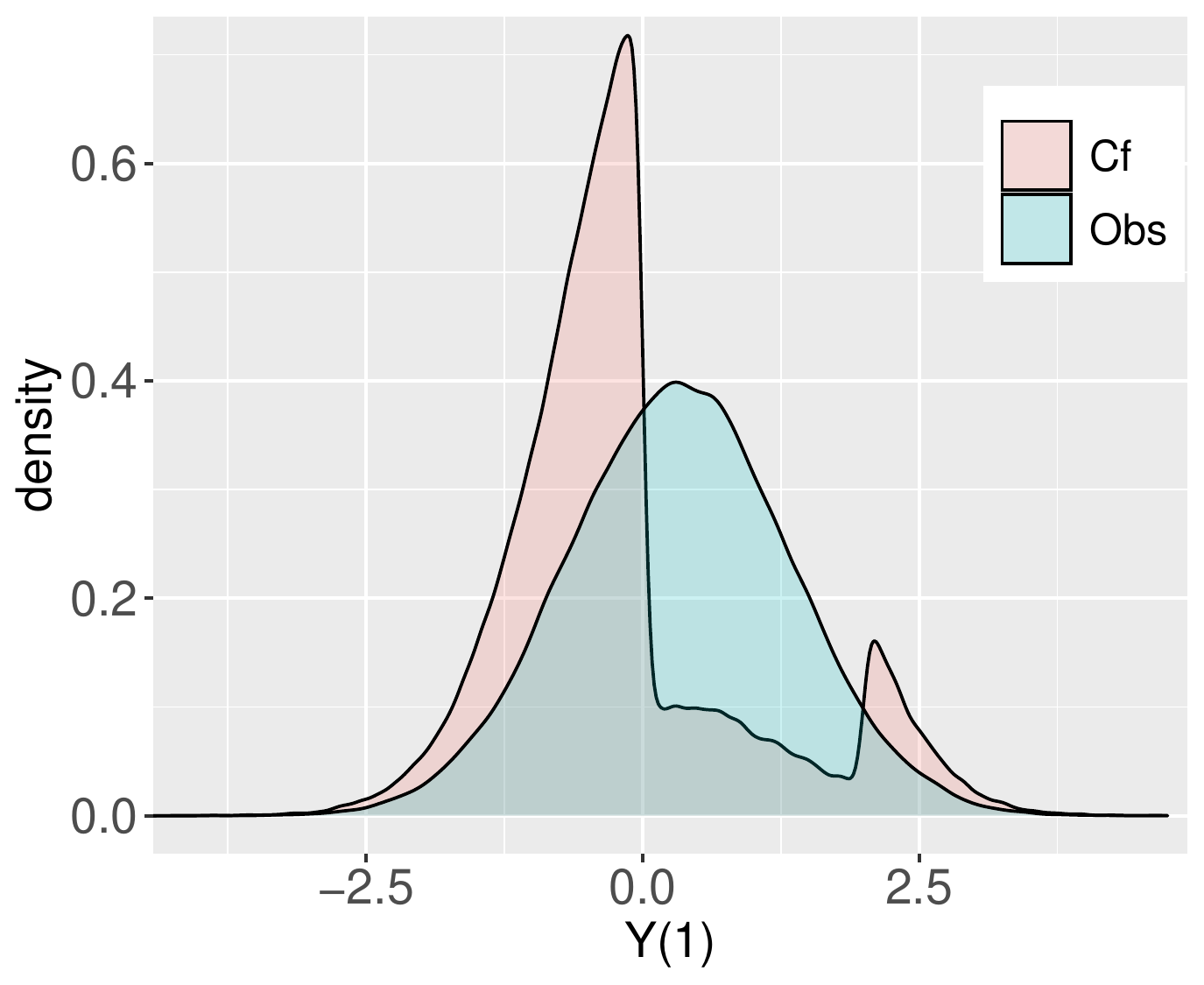} }}  
    \subfloat
    {{\includegraphics[width=0.32\textwidth]{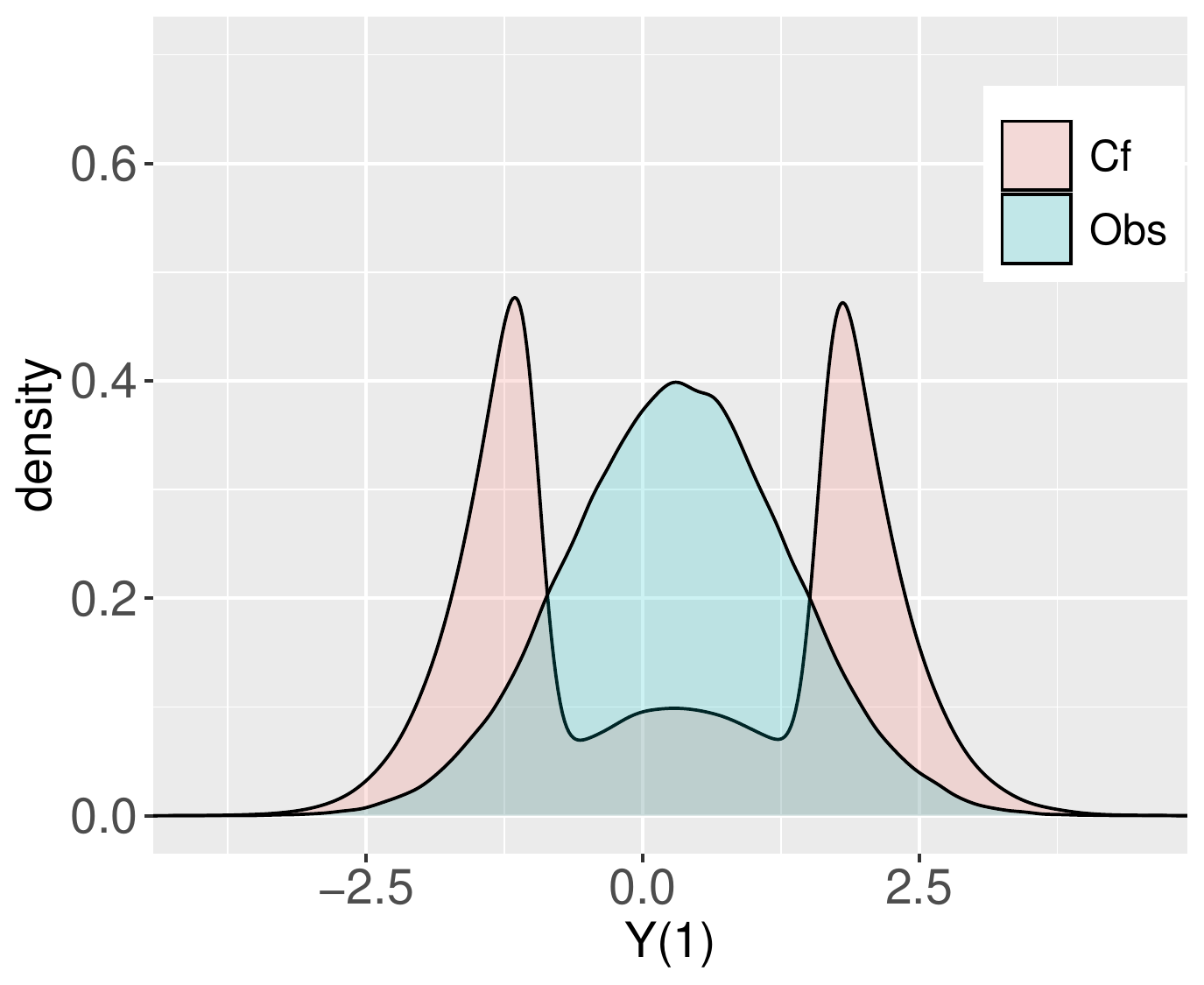} }} 
    \subfloat
   {{\includegraphics[width=0.32\textwidth]{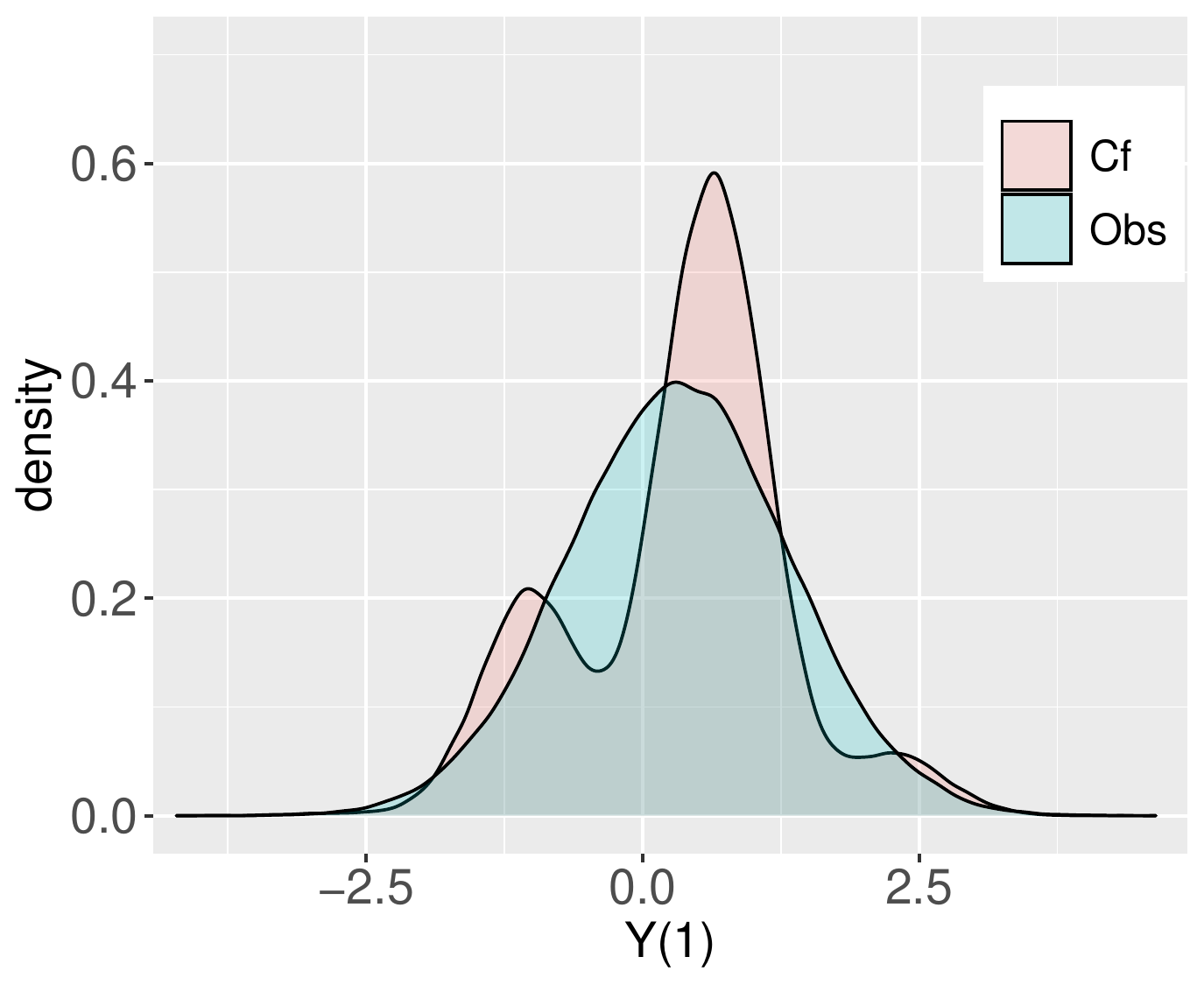} }} 
    \caption{Distribution of $Y(1)$ for the synthetic data.  Given the covariates $X$, Obs denotes the distribution $p(Y(1) \g T=1,X)$ in the observed group. Cf denotes the distribution of counterfactual outcome $p(Y(1) \g T=0,X)$. The three plots correspond to different sensitivity models in $\cE(\Gamma)$. 
    }
    \label{fig:sample}%
 }
\end{figure}

We first assume the baseline outcome $Y(0) \equiv 0$. The estimation of ITE then reduces to estimating a single potential outcome $Y(1)$.  \Cref{fig:sample} demonstates several counterfactual distributions $p(Y(1) | X, T=0)$ that are generated by the rejection sampling method described in \Cref{sec:practical} with sensitivity models in $\cE(\Gamma),~ \Gamma =4$. Unmeasured confounding is reflected as the difference between $p(Y(1) | X, T=0)$ and $p(Y(1) | X, T=1)$.  \Cref{fig:sample} show that the nonparametric  MSM  probes a variety of potential violations to unconfoundedness. %
We find the counterfactual distribution in the middle of \Cref{fig:sample}  results in the lowest coverage among the counterfactual cases in  \Cref{fig:sample} due to the mismatch of the high density regions between the observed  and counterfactual distributions. 
Since the interval estimate by CSA has coverage guarantee for any  sensitivity model in $\cE(\Gamma)$, we report the results with counterfactual in the middle of \Cref{fig:sample} as an adversarial case to test the validity of CSA.

In \Cref{sec:additional-sim} \cref{tab:coverage}, we compare the interval estimates of CSA with those  by a sensitivity analysis of the ATE \citep{zhao2019sensitivity}  and an estimation of the CATE \citep{chipman2010bart}. The (sub)population-based interval estimates underestimate the individual-level uncertainty and undercover the true ITE.  This validates the necessity of individual-level sensitivity analysis.

\Cref{fig:shrink} illustrates the properties of the estimates by CSA and CSSA. 
The top panels show  the empirical coverage under different confounding strengths. The empirical coverage is computed as $(\sum_{i=1}^m \mathbbm{1}{[Y_i(1) \in \hat{C}_1^{\Gamma}(X_i)]})/m$ for $m=10,000$ target points. We observe that ITE-NUC achieves the target coverage under unconfoundedness ($\Gamma=1$) but its coverage decreases as the confounding strength increases. In contrast, CSA-M and CSA-Q have valid coverage across all levels of $\Gamma$. The coverage of CSSA-M is above the nominal level and is lower than that of CSA-M, which demonstrate its validity and sharpness. For $\Gamma \leq 2.5$,  CSSA-M has coverage centered at the nominal level which suggests its sharpness.

The middle panels in \Cref{fig:shrink} show the average interval length on the target units. We observe that the length of ITE-NUC remains the same as $\Gamma$ changes.  In comparison, the length of  CSA and CSSA methods scales up with $\Gamma$, reflecting an increased  uncertainty with stronger unmeasured confounding. On average, $\text{CSA-M}$ produces shorter intervals than CSA-Q when the data is homoscedastic, and they have similar interval lengths when the data are heteroscedastic.  $\text{CSSA-M}$ creates shorted interval than $\text{CSA-M}$ for all $\Gamma>1$.

\begin{figure}
	\vspace*{-8mm}
	\centering{
		{{\includegraphics[width=0.47\textwidth]{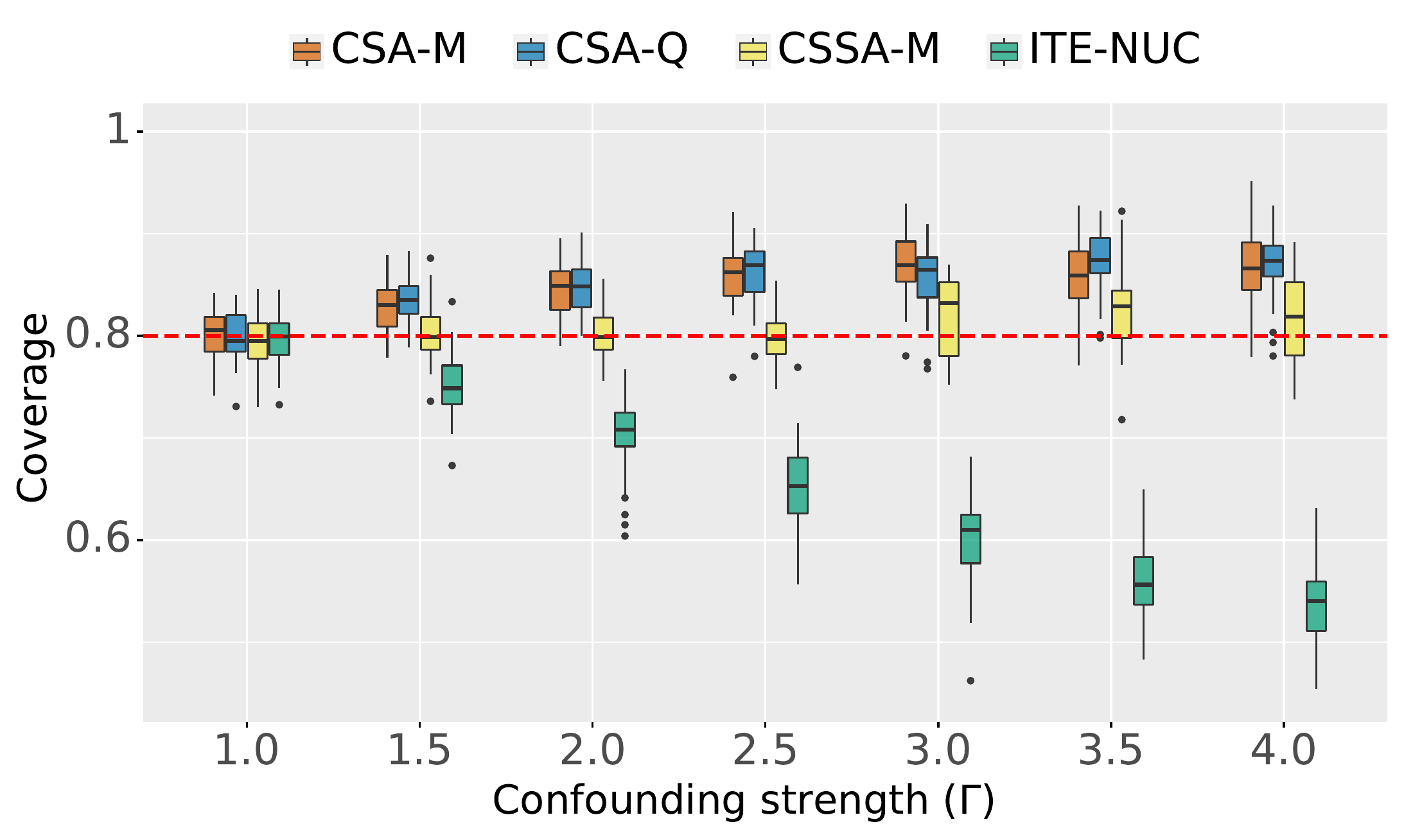} }} ~
		{{\includegraphics[width=0.47\textwidth]{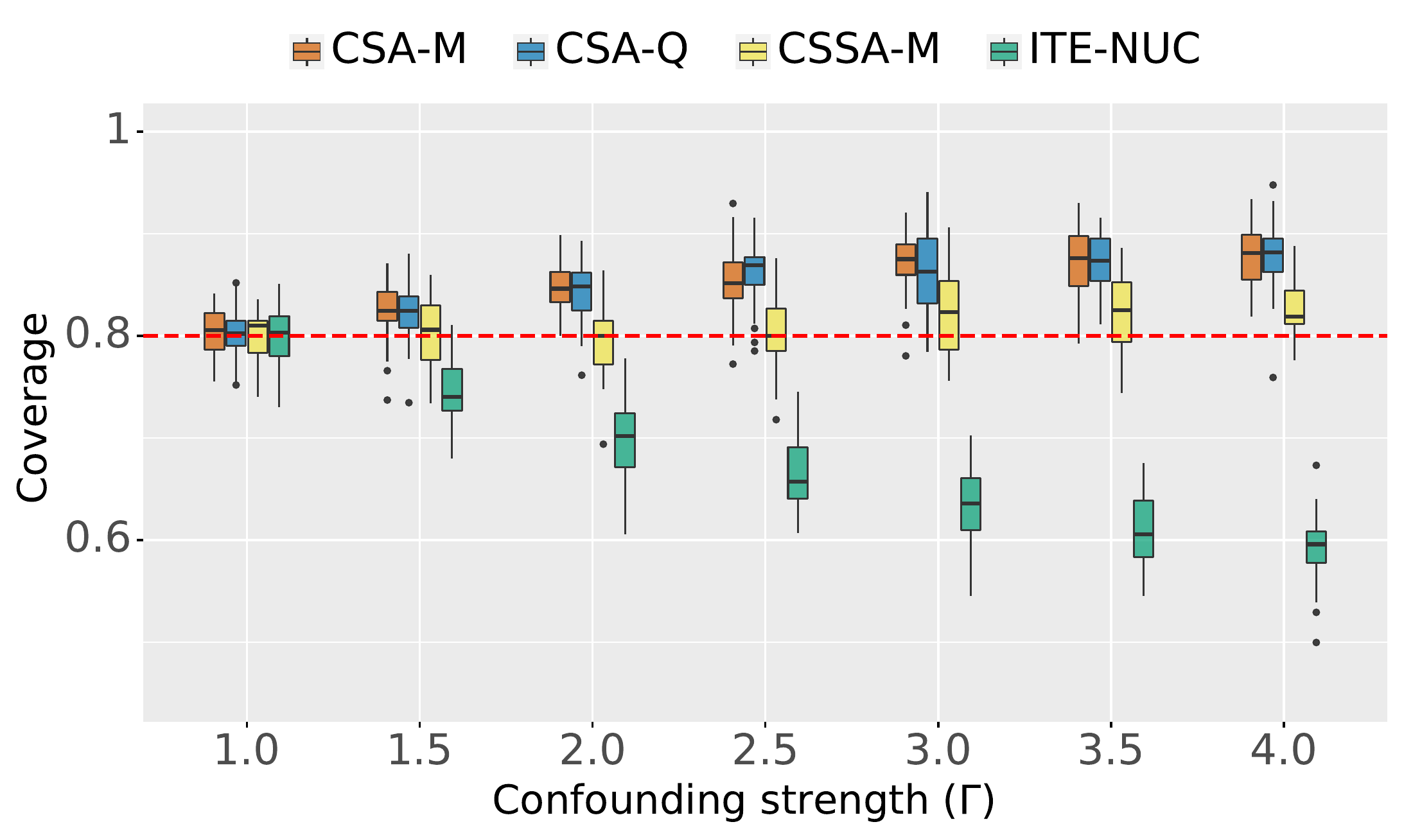} }}  \\  \vspace{3mm}
		
		{{\includegraphics[width=0.47\textwidth]{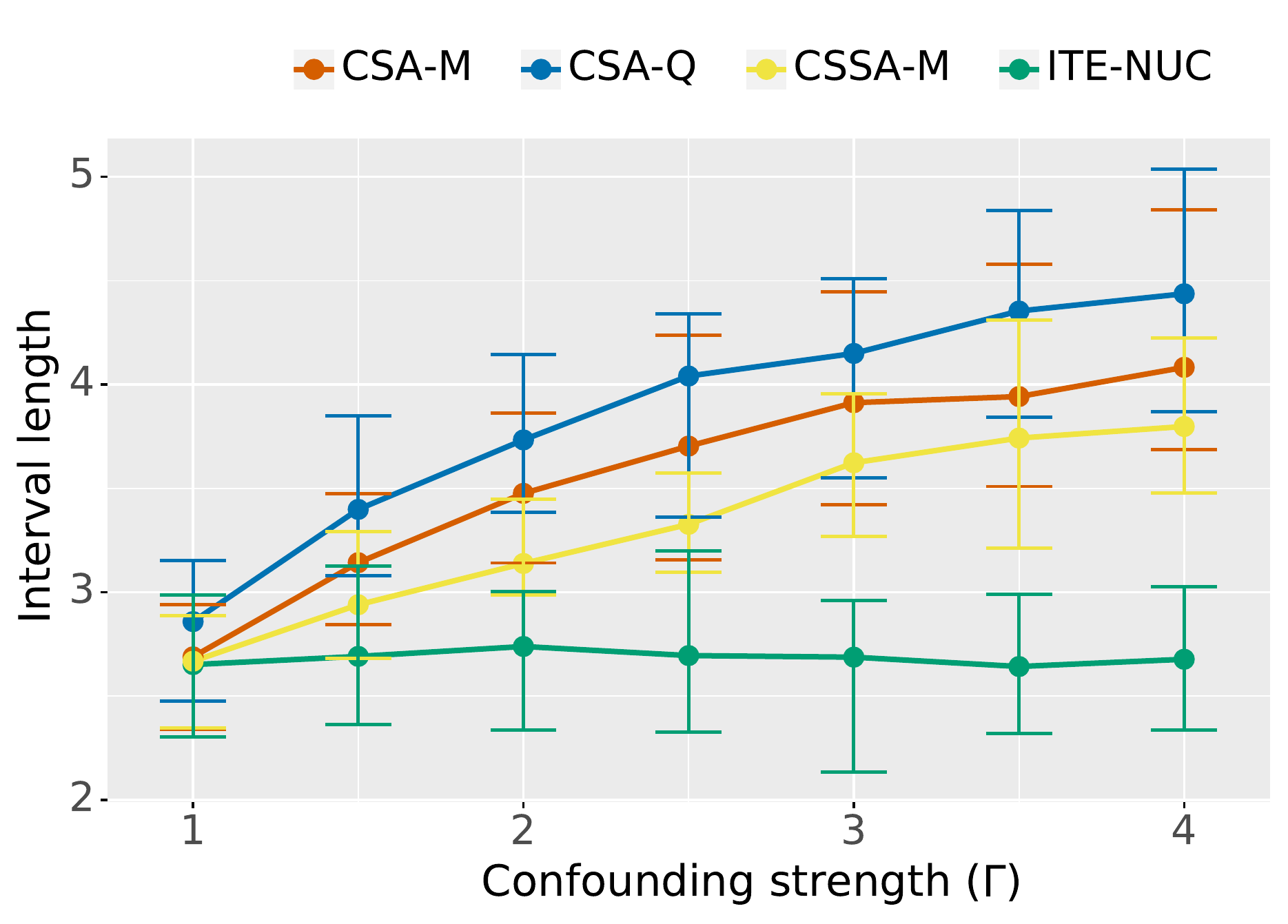} }} ~
		{{\includegraphics[width=0.47\textwidth]{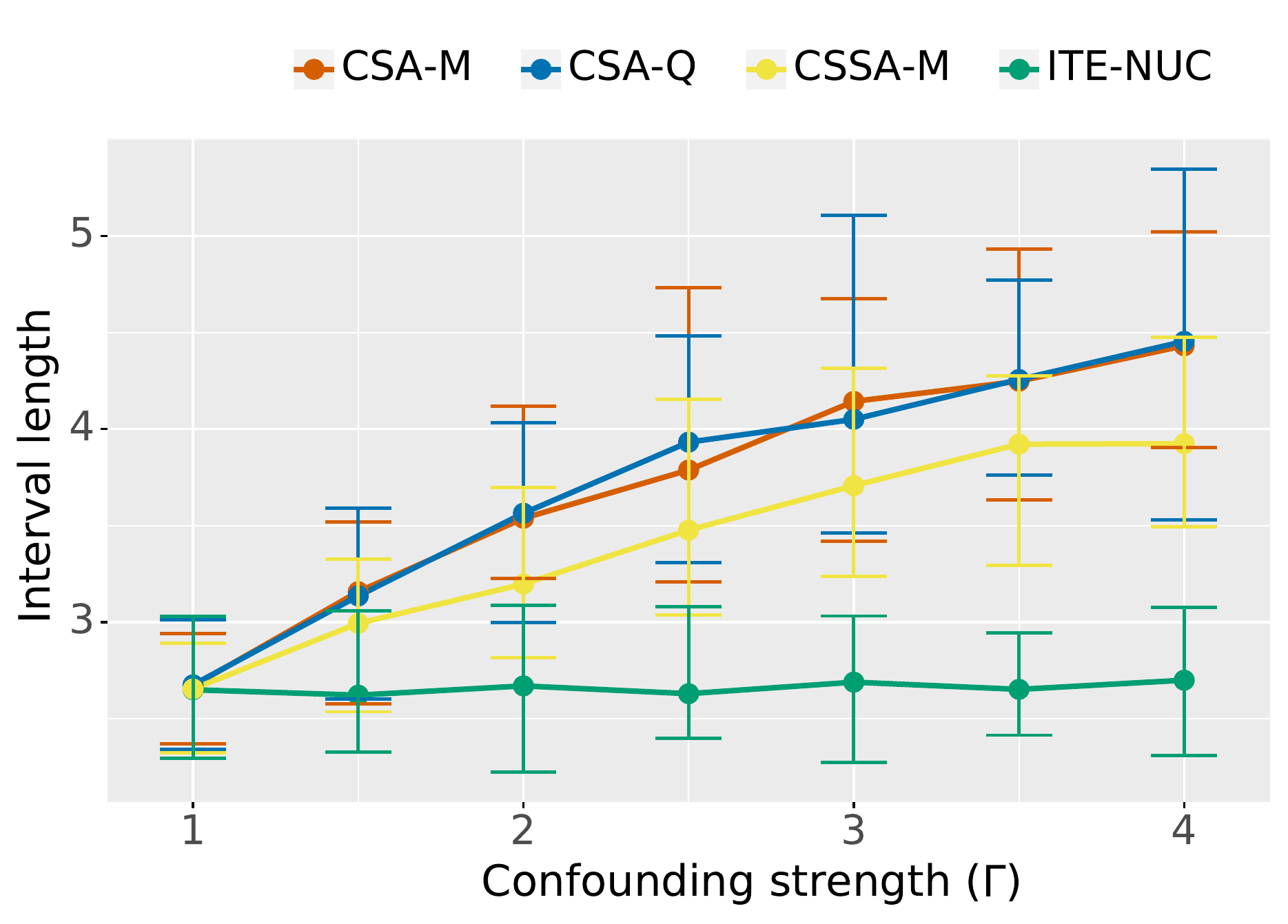} }}  \\
		
		\subfloat[Homosc.]
		{{\includegraphics[width=0.47\textwidth]{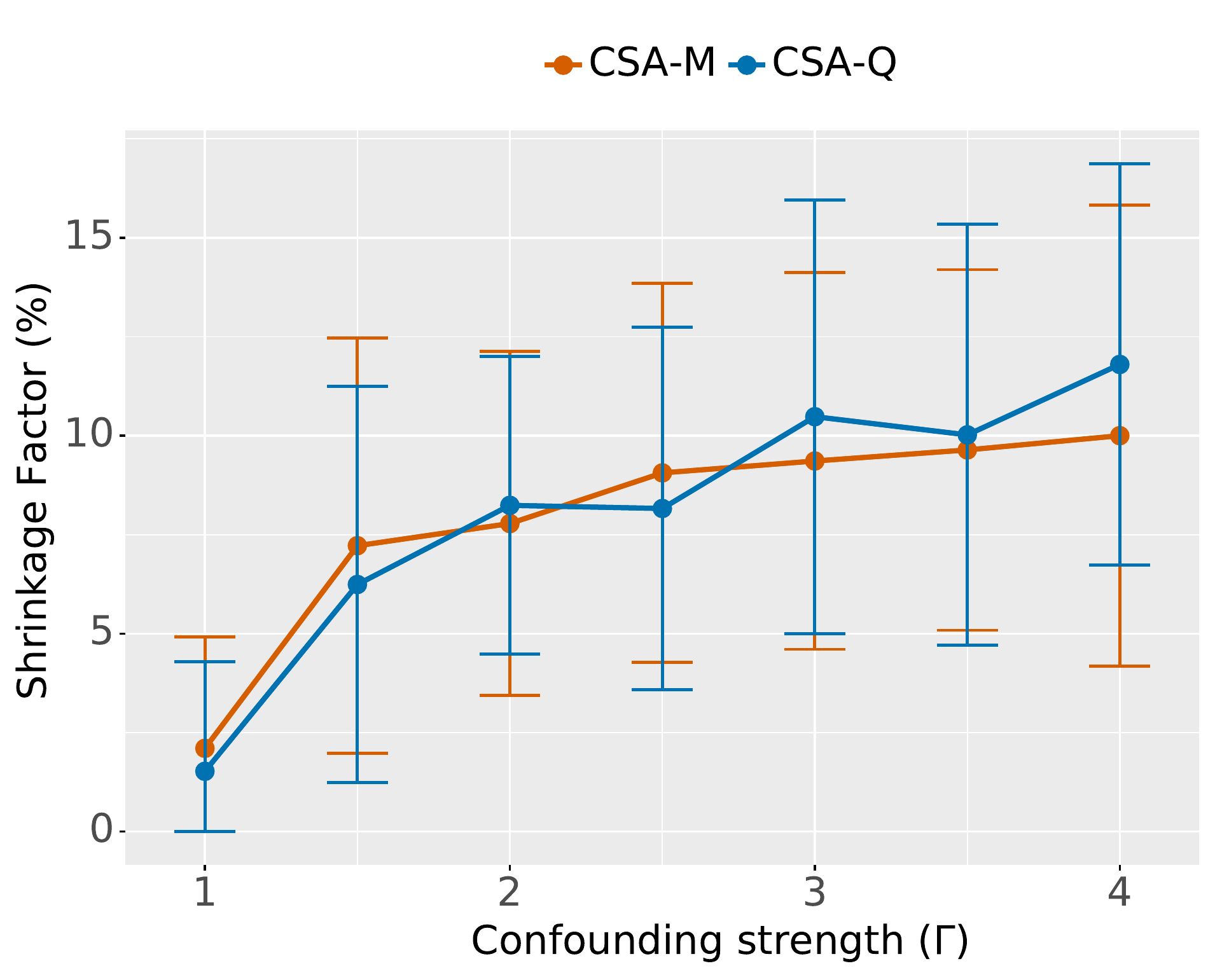} }}  ~~
		\subfloat[Heterosc.]
		{{\includegraphics[width=0.47 \textwidth]{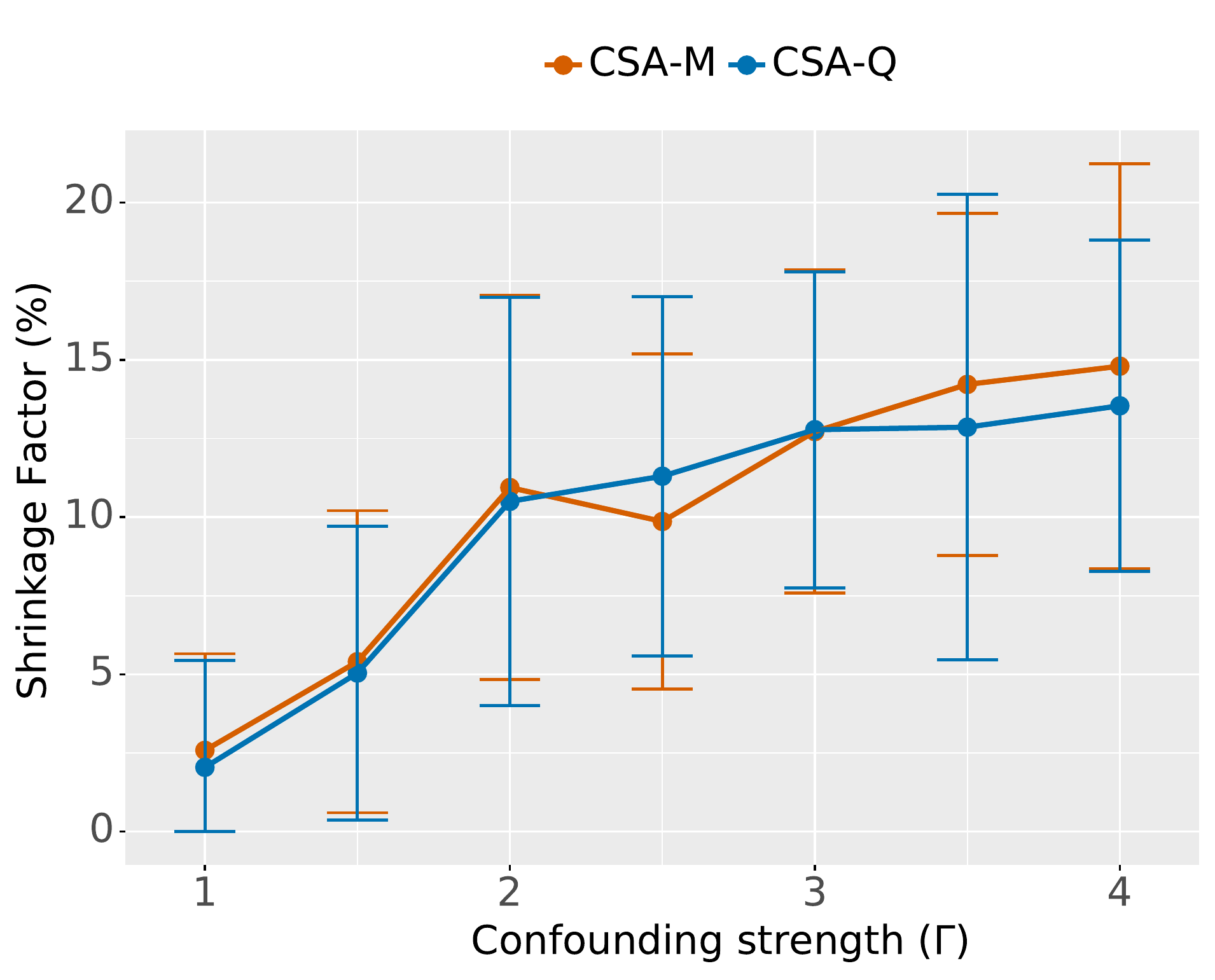} }} 
		\caption{  \small \textbf{Top}: Empirical coverage of the ITE. Dashed line denotes the target coverage level. Across $\Gamma$, the predictive intervals produced by CSA and CSSA methods reach the valid coverage.  CSSA-M improves sharpness over CSA-M. \textbf{Middle}: The average length of the predictive intervals. The interval lengths  by CSA and CSSA methods increase with $\Gamma$, reflecting increased uncertainty under unmeasured confounding. \textbf{Bottom}: The sharpness of CSA.  The maximal shrinkage factor that preserves the nominal coverage is low, which suggests  CSA methods are relatively sharp. The error bar is by 100 independent trials.}
		\label{fig:shrink}
	}
\end{figure}

To further analyze the sharpness of CSA prediction, we manually shrink  the length of the predictive intervals by a constant factor for all the units and keep the interval centers unchanged. The maximum shrinkage factor without losing the target coverage reflects the sharpness. %
From the bottom panels in \Cref{fig:shrink}, we observe that the empirical coverage drops below the $1-\alpha$ level if the shrink factor is above 10\% and 15\% for homoscedastic and heteroscedastic data, respectively. %
The  maximal shrinkage factor being low means CSA methods produce relatively sharp intervals. %

In \Cref{fig:individual}, we visualize the  ITE estimates for multiple individuals. For each unit $i$, we compute the difference between the predictive interval and the true ITE as $\hat{C}^{\Gamma}(X_i) - \tau_i$ which contains 0 if and only if $\hat{C}^{\Gamma}(X_i)$ contains $\tau_i$. For each method, we consider two confounding strengths $\Gamma \in \{1,3\}$, set the coverage $1-\alpha=0.8$, and randomly sample 70 units. When there is unmeasured confounding, ITE-NUC produces a large fraction of intervals that do not contain the ITE, but CSA methods have a small fraction of undercovered intervals on average (less than $\alpha = 0.2$) for both confounding strengths.

\subsection{CSA for the ITE Estimation}
\label{sec:simulation-ite}
We further study when both $Y_i(1)$ and $Y_i(0)$ of a unit are unobserved. The outcome $Y(1)$ is generated according to \Cref{eq:exp1y1} and the observed outcome $Y(0)$ is generated by  ~\looseness=-1
\bac{
&Y_i(0) = \E[Y_i(0) \g T_i=0,X_i] + \epsilon_i, ~\epsilon_i \sim \cN(0,\sigma^2); \\
&\E[Y_i(0) \g T_i=0,X_i] = f(X_{i1})f(X_{i2}) + \frac{10\sin(X_{i3})}{1 + \exp(-5X_{i3})}, 
}
where $f(x)$ follows the definition in \Cref{eq:exp1y1}.   The construction of the counterfactual distribution $p(Y(0) | X=x, T=0)$  is similar to the single missing outcome case, the details of which are in \Cref{sec:additional-sim}.  We analyze  the Bonferroni correction and  the  nested approach wrapped around CSA-M, CSA-Q and ITE-NUC.  For the nested approach, following \citet{lei2020conformal}, we learn the mapping  $X \mapsto \hat{C}^{\Gamma}(X)$ in \Cref{sec:ite} by fitting  40\% quantile of the lower end point and 60\% quantile of the upper end point with quantile forest function in R package \texttt{grf}.  ~\looseness=-1

The results of coverage and interval length are shown in \Cref{sec:additional-sim} \Cref{fig:shrink-2}. The Bonferroni correction provides conservative interval estimates. In comparison, the interval estimation of CSA methods with the nested method are less conservative. Similar to \cref{sec:exp-missing},  ITE-NUC has poor coverage when the unconfoundedness is violated, but CSA methods have valid coverages across different levels of confounding strength.

\begin{figure}[t]
\centering{
\includegraphics[width=0.9\textwidth]{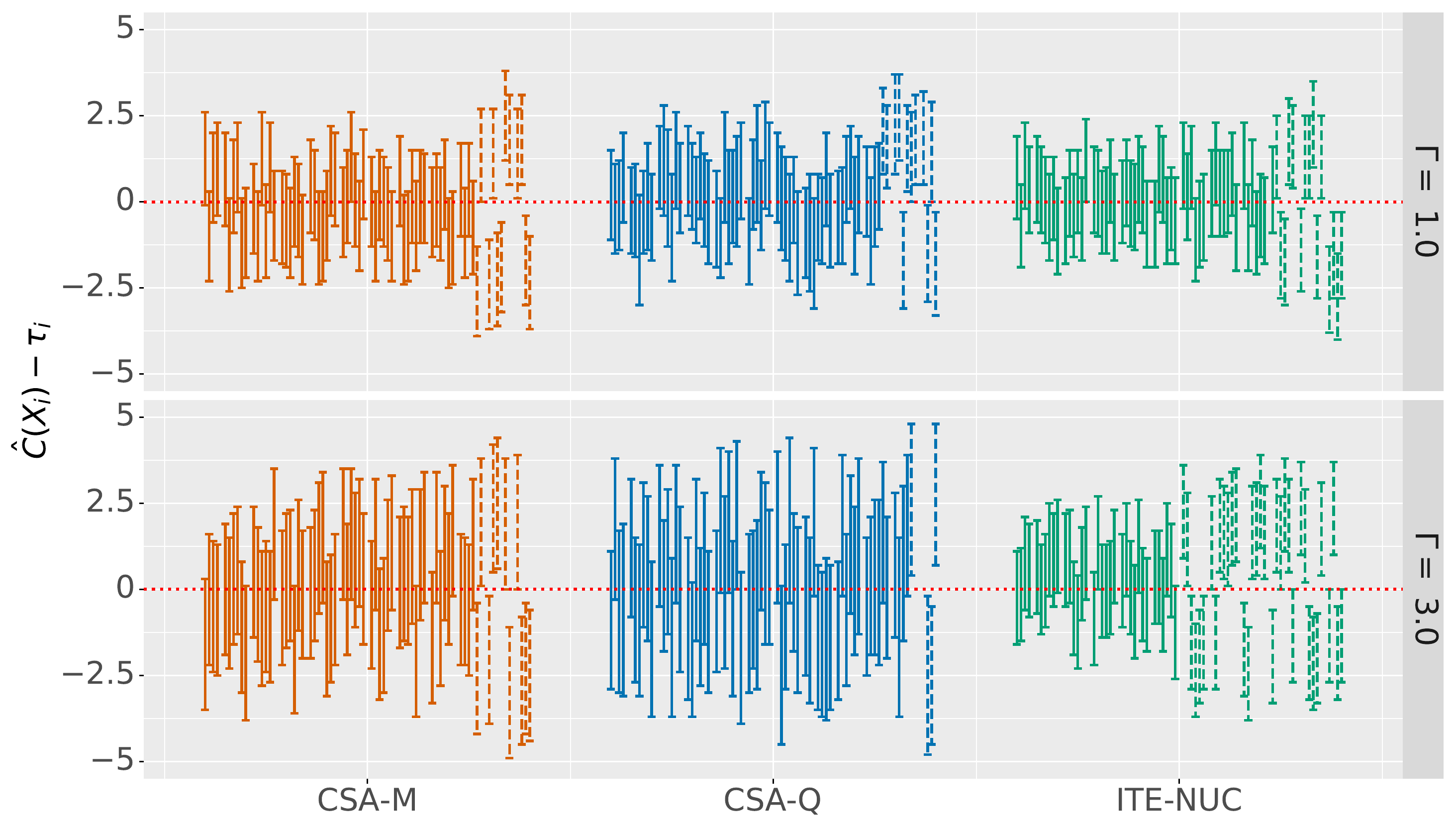} 
    \caption{The figure reports predictive intervals for random individuals with confounding strengths. %
 Each interval is the predictive interval minus the true ITE for one individual. Solid intervals contain 0 and dashed intervals do not contain 0. When $\Gamma=1$, all methods have similar coverage at $1-\alpha = 0.8$; when $\Gamma=3$, ITE-NUC has high miscoverage while CSA maintains a valid coverage.}
    \label{fig:individual}%
 }
\end{figure}

\begin{figure}[ht]
	\centering{
		\subfloat[]
		{{\includegraphics[width=0.31\textwidth]{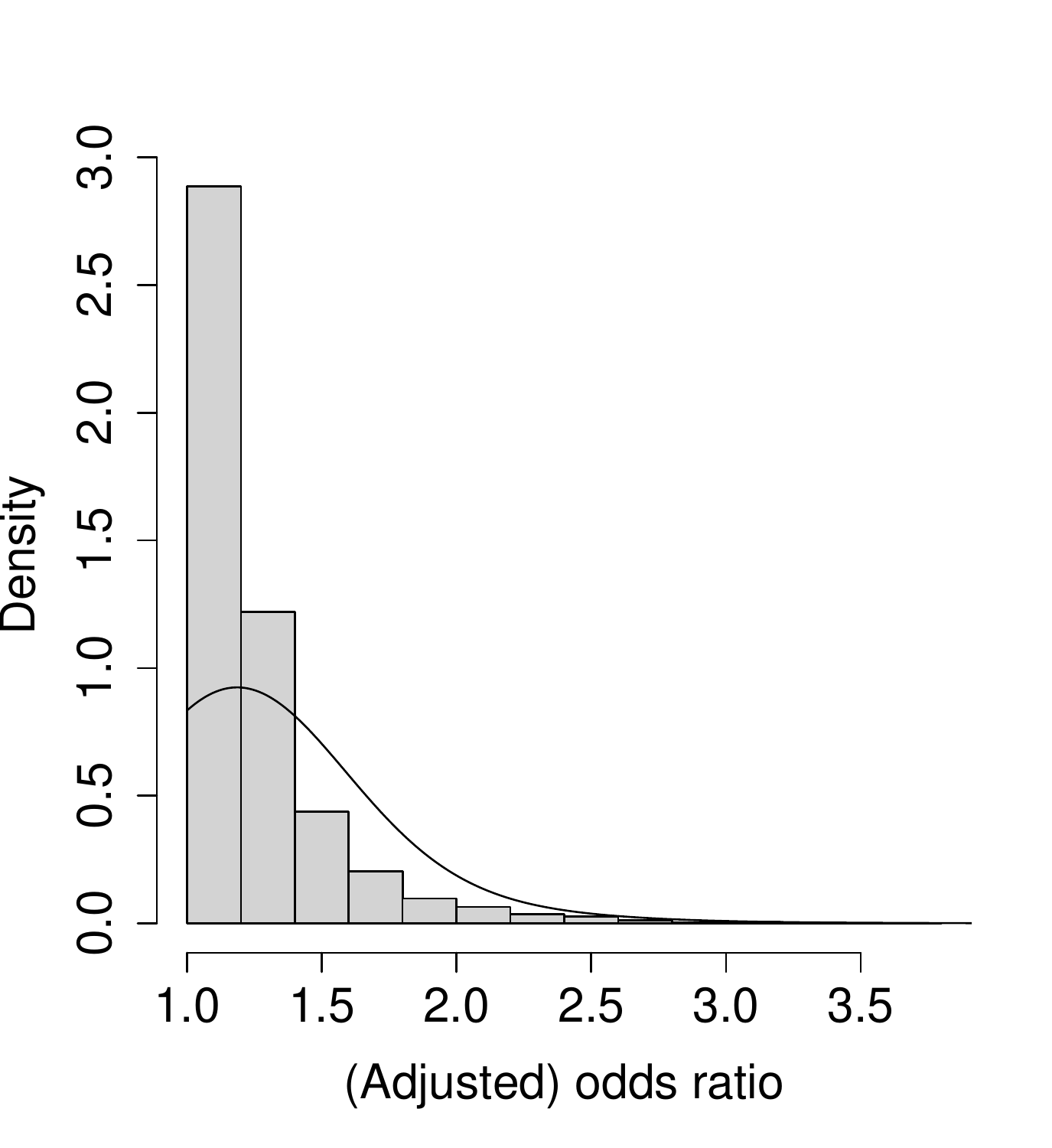} }}  \hspace*{-4mm}
		\subfloat[]
		{{\includegraphics[width=0.32\textwidth]{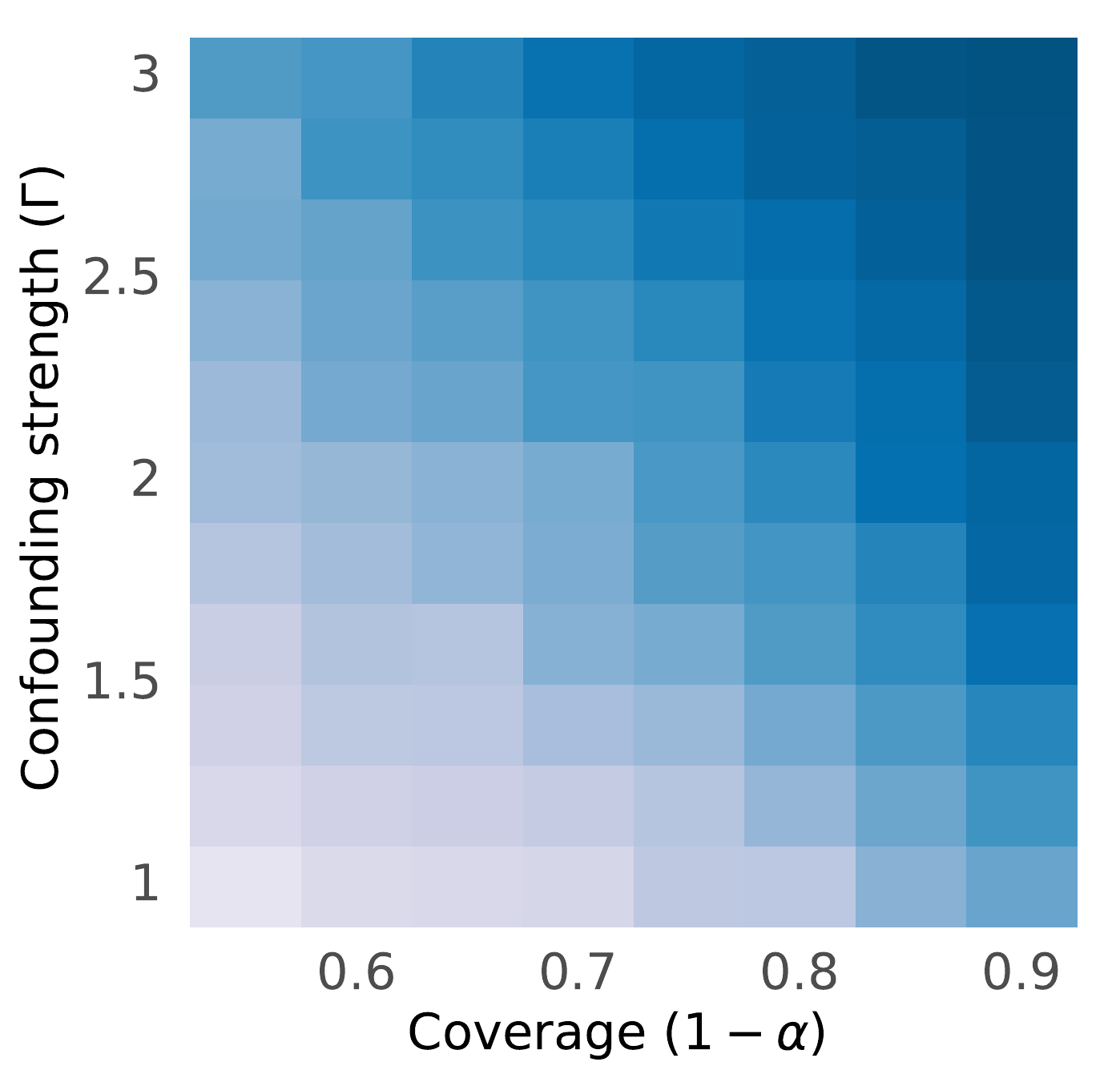} }}  ~
		\subfloat[]
		{{\includegraphics[width=0.43\textwidth]{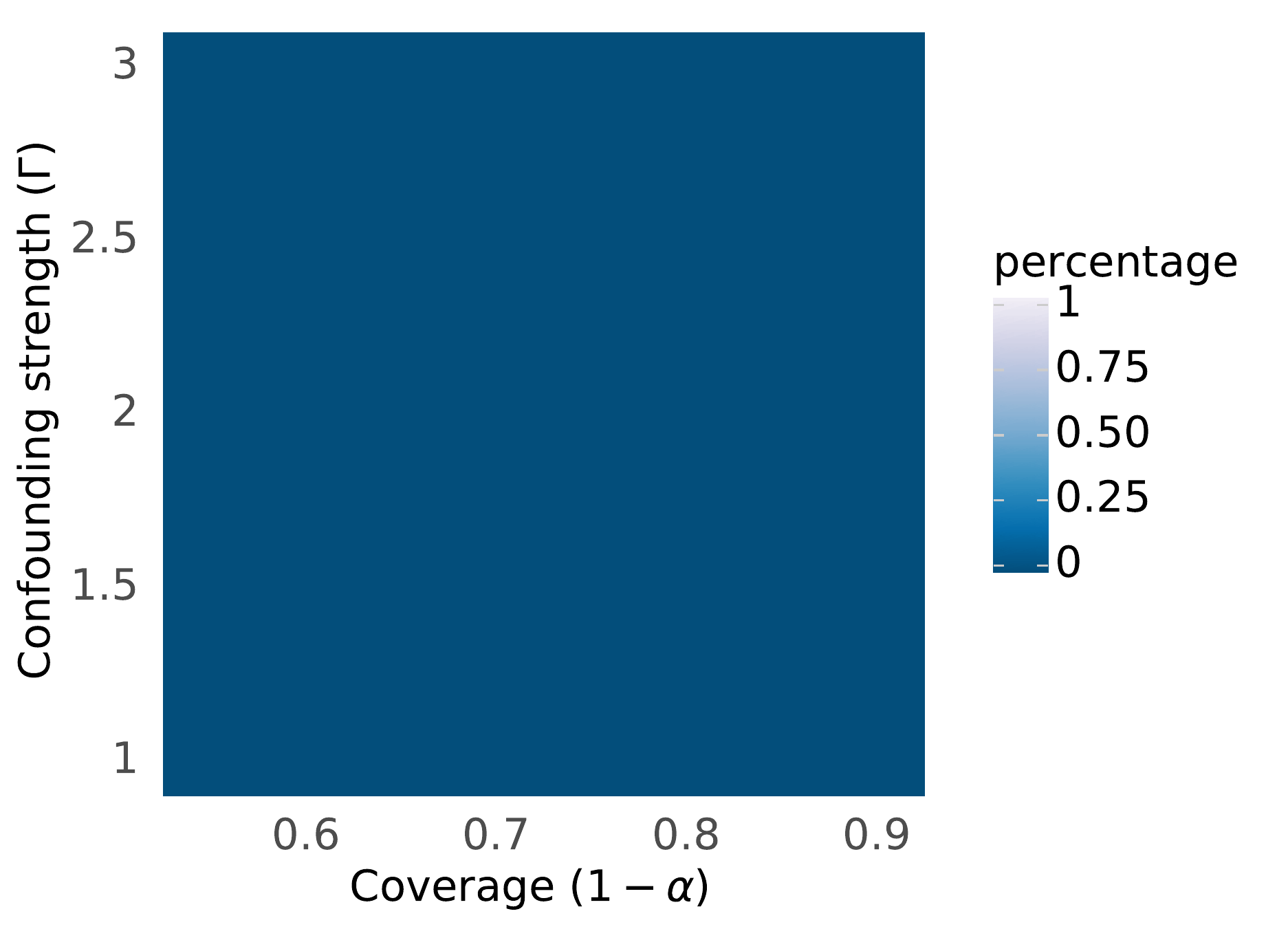} }} 
		\caption{ When the confounding strength is within the range of the study, we can say with high confidence that for a group of individuals, high fish consumption increases their blood mercury levels.
			The figures are produced using the NHANES fish consumption data. The predictive intervals are estimated by CSA-M. (a) provides reference information for the magnitude of the sensitivity parameter $\Gamma$ from the observed covariates. (b) shows the fraction of intervals with positive lower bounds; (c) shows the fraction of intervals with negative upper bounds. }
		\label{fig:fish}%
	}
\end{figure}
\subsection{Application: ITEs of Fish Consumption on Blood Mercury}

Finally, we illustrate the application of CSA  using survey responses from the National Health and Nutrition Examination Survey (NHANES) 2013-2014. The causal question we study is the effect of high fish consumption on individuals' blood mercury levels when there is potentially unmeasured confounding.

Following \citet{zhao2019sensitivity}, we define the high fish consumption as more than 12 servings of fish  a person consumes in the previous month and low fish consumption as 0 or 1 serving of fish. The outcome of interest is the blood mercury level, which is measured in ug/L and transformed to the  logarithmic scale. The dataset contains $n=1107$ units, where $80\%$ are randomly sampled as  training data and the rest $20\%$ are the target units. There are $p=8$ covariates about the demographics and health conditions \citep{zhao2018cross}.  We use random forest and quantile  forest to fit the observed outcome, the gradient boosting to estimate the propensity score and the nested method with quantile forest as the interval prediction function. %

\begin{figure}[ht]
	\centering{
		\includegraphics[width=0.9\textwidth]{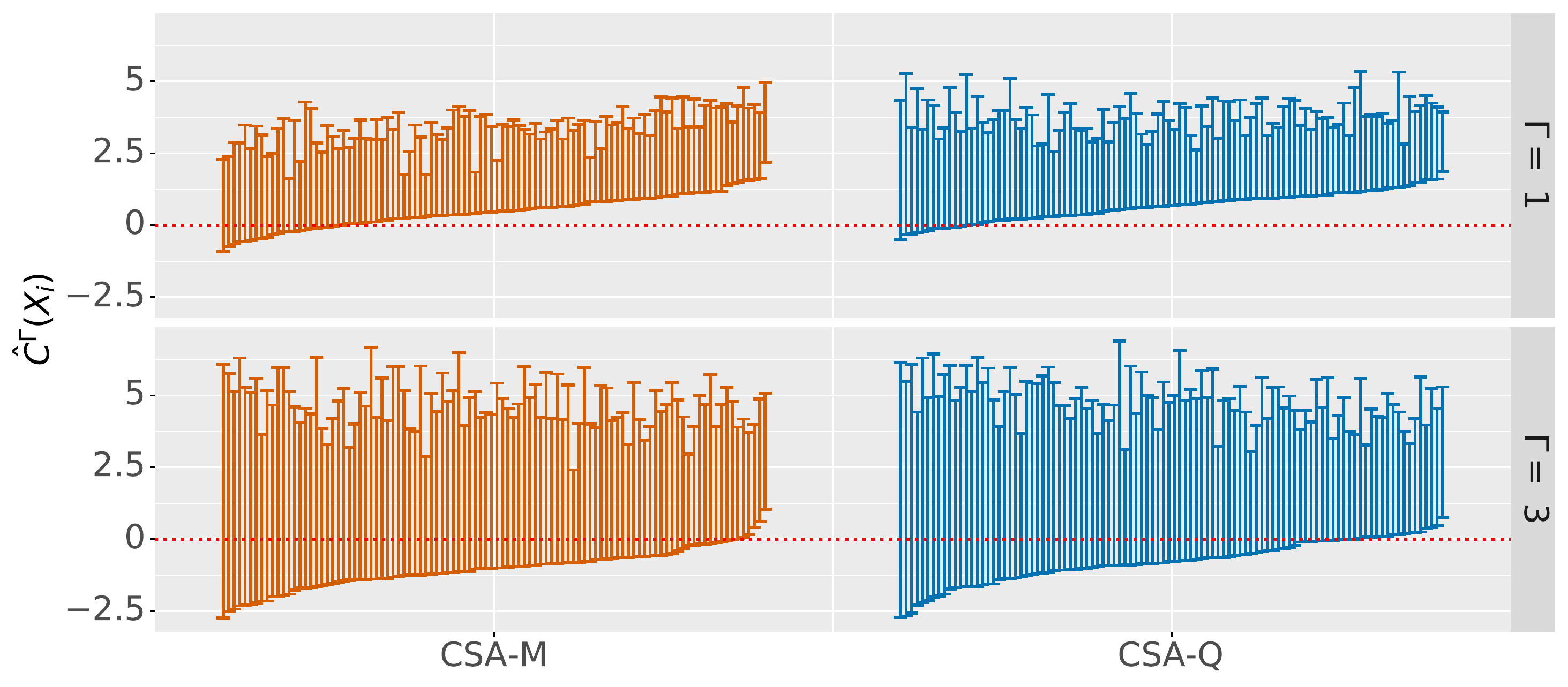} 
		\caption{\small Predictive intervals for the target individuals with different unmeasured confounding strength. In top panels, when $\Gamma=1$, a large fraction of individuals have positive effects. In the bottom panels, when $\Gamma =2$, we can still identify individuals whose effects remain positive.
		}
		\label{fig:individual-fish}%
	}
\end{figure}

We calibrate the sensitivity parameter $\Gamma$ with the observed data.  As discussed in \Cref{sec:practical}, we compute $\Gamma_{ij}$ as the effect of j-th covariate on the treatment assignment of the i-th unit in terms of odds ratio. \Cref{fig:fish} (a) shows the distribution of $\{\tilde{\Gamma}_{ij}\}_{i=1:n}^{j=1:p}$ where $\tilde{\Gamma}_{ij} $ equals to $ \Gamma_{ij}$ if $\Gamma_{ij}\geq 1$ and $1/\Gamma_{ij}$ otherwise. By  \Cref{fig:fish} (a), we may consider $\Gamma \in [1,3]$ as a plausible range of confounding strength. The choice of a proper sensitivity parameter often needs further domain knowledge in addition to the reference information from data.

For each target unit $k$, CSA produces an interval estimation $\hat{C}^{\Gamma}(X_k) = [l_k, u_k]$. We call $\hat{C}^{\Gamma}(X_k)$ a positive interval if $l_k > 0$, which represents a positive individual effect, and call $\hat{C}^{\Gamma}(X_k)$ a negative interval if $u_k < 0$.  %
\Cref{fig:fish} reports the fraction of positive and negative intervals in the target units against the target coverage $1-\alpha$ and the sensitivity parameter $\Gamma$. 
Overall, the fraction of positive intervals increases when the confounding strength and the target coverage decrease. There is a relatively strong evidence of positive effects when $\alpha \leq 0.2$ and $\Gamma \leq 2$, and there is no evidence of negative effects for $\alpha \leq 0.5$ and $\Gamma\leq 3$.

The results of individual-level estimates are reported in \Cref{fig:individual-fish}.  We randomly sample 70   individuals in the target set and show the predictive intervals of their ITEs with target coverage $1-\alpha=0.8$. In \Cref{fig:individual-fish}, the interval prediction of the treatment effect is heterogeneous across individuals. Under unconfoundedness, the ITEs for most individuals are likely to be positive. When  $\Gamma=2$, the effects of fish consumption  for some individuals are explained away by the unmeasured confounding. From \Cref{fig:individual-fish}, we can tell the subgroup for whom the effect of fish consumption on the blood mercury level is relatively insensitive to the unmeasured confounding.  The predictive intervals given by the sensitivity analysis  can thus provide useful  information to guide personal decisions on fish consumption.

\section{Discussion}

In this paper we developed a sensitivity analysis method for the ITE called CSA. %
We developed CSA by extending conformal inference to distribution shift. We adopted a two-stage design to propagate the uncertainty of an unmeasured confounding to the predictive interval of the ITE. We provided theoretical guarantees on the coverage property of the predictive interval, designed CSSA to improve the sharpness of CSA, and developed a rejection sampling method to evaluate the performance in simulation.
Finally, we analyze CSA using synthetic data and demonstrate its application in an observational study.

There are many directions for  future research. We  quantified the confounding strength by the MSM. Further research could explore alternative types of sensitivity models. If the nature of confounding is known, it might be preferable to model the effect of a confounder parametrically. We can also make a sensitivity assumption on the dependency structure between potential outcomes, which may improve the sharpness. Such dependencies can, for example, be modeled by a copula \citep{franks2019flexible,zheng2021copula}.
 Finally, the extended conformal prediction might be used to test other untestable assumptions, such as the invariant causal mechanism \citep{peters2016causal}. ~\looseness=-1

\FloatBarrier

\singlespacing
\bibliographystyle{plainnat-abb}

\bibliography{references}

\newpage
\pagenumbering{arabic}
\setcounter{page}{1}
\FloatBarrier
\bigskip
\spacingset{1.9}
\begin{center}
{\large\bf SUPPLEMENTARY MATERIAL}
\end{center}

\FloatBarrier
\appendix
\renewcommand{\thesection}{\Alph{section}}

\section{Proofs}

\begin{proof}[Proof of Lemma \ref{lm:sscore}]
By \Cref{eq:weight-raw}, we have
\ba{
w_t(X,y) =& \frac{p(x)p(Y(t)=y \g X)}{p(X \g T=t)p(Y(t) \g X,T=t)} \notag \\
=& \frac{p(T=t)}{p(T=t  \g  X)} \Big[ p(T=t \g X) + \frac{p(Y(t)=y \g X,T=1-t)}{p(Y(t)=y \g X,T=t)} \cdot p(T=1-t  \g  X) \Big].
\label{eq:lm-aux1}
}
Direct computation from \Cref{eq:sscore} shows
\ba{
\Big( \frac{1-s_t(X,y)}{s_t(X,y)} \frac{e(X)}{1-e(X)}\Big)^{2t-1} =  \frac{p(Y(t)=y \g X,T=1-t)}{p(Y(t)=y \g X,T=t)}. 
\label{eq:lm-aux2}
}
Plug \Cref{eq:lm-aux2} to \Cref{eq:lm-aux1}, we have
\bas{
w_t(X,y) =& p(T=t) \Big[ 1 +  \Big( \frac{1-s_t(X,y)}{s_t(X,y)}\Big)^{2t-1}  \Big] \\
=& p(T=t) \frac{1}{s_t(X,y)^t(1-s_t(X,y))^{1-t}} \\
\propto & \frac{1}{p(T=t \g X,Y(t)=y)}.
}
By Tukey's factorization in \Cref{eq:sscore}, it is straightforward to obtain the following equality
\ba{
\frac{p(Y(t)=y | X=x, T=1)}{p(Y(t)=y | X=x, T=0)} = \text{OR}(s_t(x,y), e(x)).
\label{eq:equiv}
}
Therefore by the MSM, $1/\Gamma \leq p(Y(t)=y | X=x, T=1)/p(Y(t)=y | X=x, T=0) \leq \Gamma$. ~\looseness=-1
\end{proof}

\begin{proof}[Proof of Lemma \ref{lm:range}]

By \Cref{eq:sscore}, when $t=1$, we have
\bas{
s_1(X,y)^{-1} =& ~1+ \frac{1-e(X)}{e(X)} \frac{p(Y(t)=y \g X,T=0)}{p(Y(t)=y \g X,T=1)}  \\
\in & ~\big[1 + \frac{1-e(X)}{\Gamma e(X)}, ~1 + \Gamma \frac{1-e(X)}{ e(X)}\big], 
}
where the interval is obtained from the density ratio bound in \Cref{eq:y-model}. 

Similarly, for $t=0$,
\bas{
s_0(X,y)^{-1} =& ~1+ \frac{e(X)}{1-e(X)} \frac{p(Y(t)=y \g X,T=1)}{p(Y(t)=y \g X,T=0)}  \\
\in & ~\big[1 + \frac{e(X)}{\Gamma (1-e(X))}, ~1 + \Gamma \frac{e(X)}{1-e(X)}\big].
}
With the above ranges of selection score, by  \Cref{lm:sscore}, we get
\bas{
w_t(x,y) \in \Big[\Big(1+\frac{1}{\Gamma }\big(\frac{1-e(x)}{e(x)}\big)^{2t-1}\Big)p(T=t),~ \Big(1+\Gamma\big(\frac{(1-e(x))}{ e(x)}\big)^{2t-1}\Big)p(T=t) \Big].
}
\end{proof}

\begin{proof}[Proof of Proposition \ref{prop:union}]
By \Cref{def:interval}, for true data generating distribution with $s_t^* \in \mathcal{E}(\Gamma)$
\bas{
\bP_{(X,Y(t))\sim p(X)p^{(s_t^*)}(Y(t) \g X)}(Y(t) \in [L,U]) \geq & \bP_{(X,Y(t))\sim p(X)p^{(s_t^*)}(Y(t) \g X)}(Y(t) \in [L^{(h_0)},U^{(h_0)}]  \\
\geq & 1-\alpha,
}
where the first inequality is true because $[L^{(s_t^*)},U^{(s_t^*)}] \subset [L,U]$.
\end{proof}

\begin{proof}[Proof of Proposition \ref{prop:qt-opt}]
Recall that the threshold m is the maximum value in $\{1,\cdots,n+1\}$ which satisfies: (i) $\sum_{j=m}^{n+1}p_j \geq \alpha$; (ii) $w_j = w_{hi}(X_j)~\text{for}~j \geq m$; (iii) $w_j = w_{lo}(X_j)~\text{for}~ j < m$. 

We first prove the second claim about the objective value. Suppose the conformal weights are $\hat{\wv} = (w_{lo}(X_1), \cdots, w_{lo}(X_{m-1}), w_{hi}(X_{m}), \cdots, w_{hi}(X_{n+1}))$, and the normalized weight is $\hat{p}_i = \hat{w}_i / \sum_{i=1}^{n} \hat{w}_i$. We prove that $Q_{1-\alpha}(\sum_{i=1}^n \hat{p}_i \delta_{V_i}) = V_m$. By definition
\ba{
Q_{1-\alpha}(V) =& \inf\{v: p(V\leq v) \geq 1-\alpha\}  \notag \\
=& \inf\{v: p(V > v)\leq \alpha\}. 
\label{eq:qt-def}
}
Plug in $V_m$, for $V \sim \sum_{i=1}^n \hat{p}_i \delta_{V_i}$, we have
\bas{
p(V > V_m) =& \sum_{i=m+1}^n \hat{p}_i \\
=& \frac{\sum_{i=m+1}^{n+1}{w_{hi}(X_i)}}{\sum_{i=1}^{m-1}{w_{lo}(X_i) +  \sum_{i=m}^{n+1}}{w_{hi}(X_i)}} \\
< &\frac{\sum_{i=m+1}^{n+1}{w_{hi}(X_i)}}{\sum_{i=1}^{m}{w_{lo}(X_i) + \sum_{i=m+1}^{n+1}}{w_{hi}(X_i)}} \\
< & \alpha.
}
The first inequality is because $w_{hi}(X_m) > w_{lo}(X_m) $. The second inequality is because of the definition of $m$, that is, if the inequality is not true, the maximum index satisfying (i)-(iii) is at least $m+1$, not $m$. 

On the other hand, for $\forall~v < V_m$, 
\bas{
p(V > v) ~\geq~  p(V\geq V_m)  ~=~ \sum_{i=m}^{n+1} \hat{p}_i  ~\geq~ \alpha,
}
where the last inequality is by definition of $m$. Then, by the quantile definition in \Cref{eq:qt-def},  $$Q_{1-\alpha}(\sum_{i=1}^n \hat{p}_i \delta_{V_i}) = V_m.$$

Next,  we prove $\hat{\wv}$ is the optima by contradiction. Suppose $\exists~w_{1:n+1}$ such that $Q_{1-\alpha}(\sum_{i=1}^{n+1}p_i \delta_{V_i}) = V_{M} > V_m$, where $w_{1:n+1}$ satisfy the optimization constraints, and $p_i = w_i/\sum_{i=1}^{n+1}w_i$. Since $V_{1:n+1}$ are ordered, we have $M > m$. By definition of quantile, we have $\sum_{i=M}^{n+1}p_i > \alpha$. Then
\bas{
\alpha <& \sum_{i=M}^{n+1}p_i \\
=& \frac{\sum_{i: i \geq M} w_i}{\sum_{i: i<M} w_i + \sum_{i: i\geq M}w_i} \\
\leq & \frac{\sum_{i: i \geq M} w_i}{\sum_{i: i<M} w_{lo}(X_i) + \sum_{i: i\geq M}w_i}  \\
\leq & \frac{\sum_{i: i \geq M} w_{hi}(X_i)}{\sum_{i: i<M} w_{lo}(X_i) + \sum_{i: i\geq M}w_{hi}(X_i)} 
}
The second to last inequality is because of the constraints $w_i \geq w_{lo}(X_i)$, and the last inequality is because of the fact that $y = x/(a+x), ~a>0$ is a monotonically increasing function of $x$ for $x>0$. Therefore, the maximum value satisfying (i)-(iii) is at least $M$ with $M>m$, which contradicts with the definition of $m$ being the maximum.  
\end{proof}

\begin{proof}[Proof of  Lemma \ref{thm:condition}]

For completeness, we first re-introduce the weighted exchangeability from \citet[Definition~1]{tib2019conformal}.  

\begin{definition} 
The variables $Z_{1:n}$ are weighted exchangeable if the joint distribution factorizes 
\ba{
f(Z_{1:n}) = \big(\prod_{i=1}^n \pi_i(Z_i)\big) \cdot g(Z_{1:n}),
} 
where $g(Z_1,\cdots, Z_n) = g(Z_{\sigma(1)},\cdots, Z_{\sigma(n)})$ for any permutation $\sigma(\cdot)$, and $\pi_i(\cdot)$ are adjusting functions. 
\label{def:exchange}
\end{definition}

We prove the result for estimated propensity score $\hat{e}(X)$. The result for $ \hat{e}(X)=e(X)$ follows as a special case. For notational convenience, we first consider the prediction band for the outcome $Y(1)$ if the unit were treated; we suppress the superscript that involves the sensitivity model since here we always consider the DGP under a fixed sensitivity model.  Denote \ba{h(x,y) \vcentcolon  = \log(\text{OR}(s_t(x,y), e(x)))= \text{logit}(s_t(x,y)) - \text{logit}(e(x)). \label{eq:h-model}} By the assumption of the MSM, $|h(x,y)| \leq \log \Gamma$. 

Consider the preliminary data $\Z_{\text{pre}}$ and the calibration data $\Z_{\text{cal}} = Z_{1:n}$ that are i.i.d. data points sampled from $p(X \g T=1) p(Y(1) \g X,T=1)$, and the target data $Z_{n+1} \sim Q(X,Y(1)) = p(X) p(Y(1) \g X)$. By \Cref{lm:sscore,eq:sscore,eq:h-model}, the  weights
\ba{
w_1(x,y) =& \frac{p(T=1)}{s_1(x,y)} \notag \\
 =& p(T=1) \Big(  1+ \frac{1-e(x)}{e(x)} \exp(-h(x,y)) \Big).
 \label{eq:thm1-pf-1}
}
The estimated conformal weights are 
\ba{
\hat{w}_1(x,y) = p(T=1) \big(  1+ \frac{1-\hat{e}(x)}{\hat{e}(x)} \exp (-h(x,y))\big).
\label{eq:thm1-pf-2}
} 
Redefine $\hat{w}_1(x,y)$ as $\hat{w}_1(x,y)/\E[\hat{w}_1(x,y)]$ where the expectation is over $p(X=x \g T=1) p(Y(1)=y \g x,T=1)$. Define $\tilde{Q}(X,Y(1)) = \hat{w}_1(X,Y(1))p(X \g T=1) p(Y(1) \g X,T=1)$. Since $\tilde{Q}(X,Y(1))$ is positive and integrates to 1, it is a density function. Consider $\tilde{Z}_{n+1}=(\tilde{X}_{n+1}, \tilde{Y}_{n+1}(1)) \sim \tilde{Q}(X,Y(1))$ and let $\tilde{Z}_{1:n} \overset{d}{=}  Z_{1:n}$.  By \Cref{def:exchange}, $\tilde{Z}_{1:n+1}$ are weighted exchangeable with adjusting function $(1,\cdots,1,\hat{w}_1)$. For the band $\hat{C}_1^\Gamma(x)$ in \Cref{eq:h-interval},
\bas{ 
\mathbb{P}(\tilde{Y}_{n+1}(1)\in \hat{C}_1^\Gamma(\tilde{X}_{n+1})  \g   \Z_{\text{pre}})
=~\mathbb{P}(\tilde{V}_{n+1} \leq Q_{1-\alpha}\big(   \sum_{i=1}^n \hat{p}_i \delta_{\tilde{V}_i} + \hat{p}_{n+1} \delta_{\infty}    \big)  \g  \Z_{\text{pre}}) 
}
where $\tilde{V}_i = V(\tilde{Z}_i)$, $\hat{p}_i = \hat{w}_1(\tilde{Z}_i) /\sum_{i=1}^{n+1} \hat{w}_1(\tilde{Z}_i)$.  By \citet[Lemma 3]{tib2019conformal} and its proof,  the right hand side of \Cref{eq:thm1-pf-1} is within $[1-\alpha, 1-\alpha+\max_{i}\{\hat{p}_{i}\}_{i=1}^{n+1}]$.  Therefore 
\ba{
1-\alpha \leq \mathbb{P}(\tilde{Y}_{n+1}(1)\in \hat{C}_1^\Gamma(\tilde{X}_{n+1})  \g   \Z_{\text{pre}}) \leq 1-\alpha+\max_{i}\{\hat{p}_{i}\}_{i=1}^{n+1}.
\label{eq:thm1-pf}
}

We first prove the lower bound of the coverage probability. By definition of total variation distance $d_{\text{TV}}(\cdot)$,
\ba{
\resizebox{.92\hsize}{!}{
$\left|p(Y_{n+1}(1)\in \hat{C}_1^\Gamma(X_{n+1})  \g   \Z_{\text{pre}},  \Z_{\text{cal}}) - p(\tilde{Y}_{n+1}(1)\in \hat{C}_1^\Gamma(\tilde{X}_{n+1})  \g   \Z_{\text{pre}},  \Z_{\text{cal}})\right| 
 \leq~d_{\text{TV}}(Q, \tilde{Q}),$
}
\label{eq:thm1-pf-3}
}
where $d_{\text{TV}}(Q, \tilde{Q})=d_{\text{TV}}(Q(X,Y(1)), \tilde{Q}(X,Y(1)))$.  Then
\bas{
\mathbb{P}(Y_{n+1}(1)\in \hat{C}_1^\Gamma(X_{n+1})  \g   \Z_{\text{pre}}, \Z_{\text{cal}})  \geq p(\tilde{Y}_{n+1}(1)\in \hat{C}_1^\Gamma(\tilde{X}_{n+1})  \g   \Z_{\text{pre}},  \Z_{\text{cal}}) - d_{\text{TV}}(Q, \tilde{Q}).
}
Marginalizing over $ \Z_{\text{cal}}$ and by \Cref{eq:thm1-pf},
\bas{
\mathbb{P}(Y_{n+1}(1)\in \hat{C}_1^\Gamma(X_{n+1})  \g   \Z_{\text{pre}})  \geq 1-\alpha - d_{\text{TV}}(Q, \tilde{Q}).
}
By the integral definition of total variance,
\bas{
d_{\text{TV}}(Q, \tilde{Q}) =\half \int |\hat{w}(X,Y(1)) - w(X,Y(1))| p(X \g T=1) p(Y(1) \g X,T=1). 
}
Plug in \Cref{eq:thm1-pf-1,eq:thm1-pf-2}
\bas{
|\hat{w}(X,Y(1)) - w(X,Y(1))|  =& p(T=1) |\frac{1-\hat{e}(X)}{\hat{e}(X)} - \frac{1-e(X)}{e(X)}| \exp(-h(X,Y(1))) \\
\leq & \Gamma p(T=1) |\frac{1-\hat{e}(X)}{\hat{e}(X)} - \frac{1-e(X)}{e(X)}|.
}
Hence, we have
\ba{
d_{\text{TV}}(Q, \tilde{Q})  \leq  \frac{\Gamma}{2} p(T=1) \E_{X \sim p(X \g T=1)}|\frac{1-\hat{e}(X)}{\hat{e}(X)} - \frac{1-e(X)}{e(X)}|
\label{eq:thm1-pf-tv}
}
Finally, marginalizing over $ \Z_{\text{pre}}$ gives 
\bas{
\mathbb{P}(Y_{n+1}(1)\in \hat{C}_1^\Gamma(X_{n+1})  \g   \Z_{\text{pre}})  \geq 1-\alpha - \frac{\Gamma}{2} p(T=1) \E_{X \sim p(X \g T=1)}|\frac{1-\hat{e}(X)}{\hat{e}(X)} - \frac{1-e(X)}{e(X)}|.
}

Now we prove the upper bound. Back to \Cref{eq:thm1-pf}, we first bound $\max_{i}\{\hat{p}_{i}\}_{i=1}^{n+1}$. By \Cref{eq:thm1-pf-2}  
\ba{
\hat{p}_i \leq& \frac{1+\Gamma \frac{1-\hat{e}(X_i)}{\hat{e}(X_i)}}{1+\Gamma \frac{1-\hat{e}(X_i)}{\hat{e}(X_i)} + \sum_{j\neq i}(1+\frac{1}{\Gamma} \frac{1-\hat{e}(X_j)}{\hat{e}(X_j)})} \notag \\
\leq& \frac{1+\Gamma \frac{1-\eta}{\eta}}{1+\Gamma \frac{1-\eta}{\eta} + n + \frac{1}{\Gamma} \frac{\eta}{1-\eta}} \notag \\
\leq&  \frac{\Gamma/\eta}{n+\Gamma/\eta}
\label{eq:thm1-pf-pi}
}
where the first inequality is by the MSM assumption and the second is by the bounds of propensity score. Then by \Cref{eq:thm1-pf,eq:thm1-pf-3}, we have
\ba{
p(Y_{n+1}(1)\in \hat{C}_1^\Gamma(X_{n+1})  \g   \Z_{\text{pre}})  \leq 1-\alpha+\max_{i}\{\hat{p}_{i}\}_{i=1}^{n+1} + d_{\text{TV}}(Q, \tilde{Q})
\label{eq:thm1-pf-upper}
}
Plugging the bounds in \Cref{eq:thm1-pf-tv,eq:thm1-pf-pi} to \Cref{eq:thm1-pf-upper} and marginalizing out $ \Z_{\text{pre}}$ we get the  upper bound. 

We can proceed the proof similarly to the case of $t=0$.  Hence, the proof is completed. ~\looseness=-1
\end{proof}

\section{Conformalized Sharp Sensitivity Analysis}
\label{sec:cssa}

We first show the sharpness constraint in \citet[Proposition 3, Corollary 2]{Dorn2021-vs} is equivalent to the constraint in the sharp MSM \Cref{eq:sharp-msm}. 

\begin{proposition}
The condition $\int p^{(s_t)}(Y(1)=y \g X, T=0) dy = 1$  in the sharp MSM is equivalent to $\E[\frac{T}{s_1(X,Y(1))} \g X] = 1$.
\label{prop:sharp}
\end{proposition}

\begin{proof}[Proof of \Cref{prop:sharp}]
We prove for the potential outcome under treatment and the claim for the PO under control follows similarly. By \Cref{eq:equiv}, for a given $X$, 
\bas{
\int p(Y(1)=y \g X, T=0) dy =& \int p(Y(1)=y \g X, T=1) \frac{e(X)}{1-e(X)} \frac{1}{s_1(X,y)} dy -\frac{e(X)}{1-e(X)} \\
=& \E[\frac{e(X)}{s_1(X,Y(1)} \g X, T=1] \frac{1}{1-e(X)}  -\frac{e(X)}{1-e(X)} \\
=&  \E[\frac{T}{s_1(X,Y(1))} \g X]  \frac{1}{1-e(X)}  -\frac{e(X)}{1-e(X)}
}
 where the last equality is because 
\bas{
\E[\frac{T}{s_1(X,Y(1))} \g X] =& \E[\frac{1}{s_1(X,Y(1)} \g X, T=1] e(X).
}
Therefore, we have
\ba{
\int p(Y(1)=y \g X, T=0) dy = 1 \Leftrightarrow  \E[\frac{T}{s_1(X,Y(1))} \g X] = 1.
}
\end{proof}

\noindent\textbf{Covariate Balancing Selection Score.~~}   For the  control group, similar sharpness constraint as \Cref{eq:relaxed} in main paper can be obtained by flipping the binary treatment label. By \Cref{eq:relaxed}, we get 
\ba{
\E[\frac{g(X_i)T_i}{s_1(X_i,Y_i(1))} ] = \E[\frac{g(X_i)(1-T_i)}{1-s_0(X_i,Y_i(0))} ] =  \E[g(X_i)]. 
\label{eq:balancing}
}
\Cref{eq:balancing} encourages the covariate balancing by moment matching. As a related work, the idea to tighten the interval estimation by covariate balancing has  been applied to generalizing RCT estimates across different locations \citep{nie2021covariate}. 

Note that under unconfoundedness,  i.e. $\Gamma=1$ and $s_1(X,Y(1)) = s_0(X,Y(0)) = e(X)$, \Cref{eq:balancing} recovers the constraints of the covariate balancing propensity score (CBPS) \citep{imai2014covariate}. Hence we call the selection scores that satisfies \Cref{eq:balancing} the covariate balancing selection score (CBSS). For the choice of covariates function $g(X)$, \citet{Dorn2021-vs} proposes $g(X)=Q_\tau(Y|X,T=1)$; \citet{imai2014covariate} takes $g(X)$ as  $X$ and the derivative of the   treatment model w.r.t.  its parameters. \\

\noindent\textbf{Derivation of the Constraints \Cref{eq:qt-opt-sharp} in \Cref{sec:sharpness}.~~}   
By \Cref{lm:sscore}, we re-write the balancing condition in \Cref{eq:relaxed} in terms of the conformal weights,
\ba{
 \E[g(X_i)] =& \E[\frac{g(X_i)T_i}{s_1(X_i,Y_i(1))} ] \notag \\
 =& \E[\frac{g(X_i)T w_1(X_i,Y_i(1))}{p(T_i=1)} ]  \notag  \\
=& \E[g(X_i)w_1(X,Y_i(1)) | T_i=1]. 
}
By definition of the propensity score, we also have
\ba{
\E[\frac{g(X_i)T_i}{e(X_i)} ] = \E[g(X_i)]
}
Hence by \Cref{eq:relaxed}, $ \E[g(X_i)w_1(X,Y_i(1)) | T_i=1] = \E[\frac{g(X_i)T_i}{e(X_i)} ] $. The empirical estimation to this balancing condition is the equality constraint in  \Cref{eq:qt-opt-sharp}, i.e.,
\bas{
\frac{1}{N_t}\sum_{i: T_i=t} g(X_i)&w_i^{\Gamma} = \frac{1}{N} \sum_{i=1}^N \frac{T_i g_k(X_i)}{\hat{e}(X_i)}.
}
Similar condition can be derived for the control group by flipping the treatment label. 

In sum,  CSSA solves the optimization in \Cref{eq:qt-opt} with additional constraints in \Cref{eq:qt-opt-sharp}, 
\bac{
\max_{w_{1:n+1}}  \quad &Q_{1-\alpha}\big(   \sum_{i=1}^n p_i \delta_{V_i} + p_{n+1}\delta_{\infty}  \big). \\
  \text{subject to} \quad& p_i = \frac{w_i}{\sum_{i=1}^{n+1} w_i}, ~~ 1 \leq i \leq n+1 \\
 w_{lo}^{\Gamma}(X_i) \leq w_i &\leq w_{hi}^{\Gamma}(X_i), ~1 \leq i \leq n, ~~w_{lo}^{\Gamma}(X) \leq w_{n+1} \leq w_{hi}^{\Gamma}(X) \\
\frac{1}{N_t}\sum_{i: T_i=t} g_k(X_i)&w_i^{\Gamma} = \frac{1}{N} \sum_{i=1}^N \frac{T_i^t(1-T_i)^{1-t}}{\hat{e}(X_i)^t(1-\hat{e}(X_i))^{1-t}}g_k(X_i),  ~~ 1 \leq k \leq K.
\label{eq:qt-opt-sharp1}
} 
\vspace{2mm}

For example, when $t=1$, $K=1$, let the covariate function be the estimated propensity score, i.e. $g_1(X_i) = \hat{e}(X_i)$. The balancing constraint in \Cref{eq:qt-opt-sharp1} becomes
\ba{
\frac{1}{N_t} \sum_{i:T_i=1}  e(X_i)  w_i / p(T_i=1) = 1,
}
where the propensity score $e(X)$ can be replaced by $\hat{e}(X_i)$ in practice. Note that by the range of conformal weights in \Cref{lm:range}, this balancing condition is automatically satisfied under unconfoundedness, i.e.  $\Gamma=1$, and \Cref{eq:qt-opt-sharp1} reduces to the standard WCP.  \\

\noindent\textbf{Computational Methods for Solving \Cref{eq:qt-opt-sharp1}.~~}    The difficulty for solving the quantile optimization in \Cref{eq:qt-opt-sharp1} is that the objective function is the empirical quantile, which is not continuous. To solve 
this problem, we consider a sequence of simplified optimization with continuous objective.  For a given integer $1 \leq J \leq n+1$,
\bac{
\max_{w_{1:n+1}}  \quad & \frac{\sum_{i=J}^{n+1} w_i}{\sum_{i=1}^{n+1} w_i}  \\
  \text{subject to} \quad& 
 w_{lo}^{\Gamma}(X_i) \leq w_i \leq w_{hi}^{\Gamma}(X_i), ~1 \leq i \leq n, ~~w_{lo}^{\Gamma}(X) \leq w_{n+1} \leq w_{hi}^{\Gamma}(X) \\
& \frac{1}{N_t}\sum_{i: T_i=t} g_k(X_i) w_i = \frac{1}{N} \sum_{i=1}^N \frac{T_i^t(1-T_i)^{1-t}}{\hat{e}(X_i)^t(1-\hat{e}(X_i))^{1-t}}g_k(X_i),  ~~ 1 \leq k \leq K.
\label{eq:qt-opt-simple}
}
The problem in \Cref{eq:qt-opt-simple} is a nonlinear programming with linear constraints, which can be solved efficiently with existing computation tools  (e.g., the sequential quadratic programming in Python package \texttt{scipy} and R package \texttt{nloptr}).  Denote the optima of \Cref{eq:qt-opt-simple} as $\hat{\alpha}_J$. It is clear that $1=\hat{\alpha}_1 \geq \hat{\alpha}_2 \geq \cdots \geq \hat{\alpha}_{n+1}$.  Supposing $V_1 \leq \cdots \leq V_n \leq V_{n+1}$, by the definition of quantile, the solution of \Cref{eq:qt-opt-sharp} is $V_j$ where $j = \max \{j: \hat{\alpha}_j \geq \alpha\}$. We can then solve \Cref{eq:qt-opt-simple} for $J = n+1, n, \cdots, 1$ until the first time $\hat{\alpha}_{J} \geq \alpha$ (say $J=m$) and  the optima of \Cref{eq:qt-opt-sharp1} is $V_m$. This process can be implemented efficiently with binary search because of the ordering in $\hat{\alpha}_{1:n+1}$. Hence we only need to solve \Cref{eq:qt-opt-simple} at most $\mathcal{O}(\log n)$ times. In practice, to further improve the computational efficiency, we can first run \Cref{alg:csa} to find an upper bound of the optimal index $m$ before the binary search. CSSA is summarized in \Cref{alg:cssa}. ~\looseness=-1

{\LinesNumberedHidden
\begin{algorithm}[!t]  
  \caption{CSSA for Estimating a Potential Outcome}
  \label{alg:cssa}
\textbf{Input:} Data $\Z=(X_{i}, Y_{i}, T_{i})_{i=1}^{N}$, where $Y_{i}$ is missing if $T_{i}=1-t$; level $\alpha$, sensitivity parameter $\Gamma$, target point covariates $X$, covariate balancing functions $\{g_k(X)\}_{k=1}^K$
 \vspace{0.2cm}
\begin{algorithmic}[1]
 \State  Run \Cref{alg:csa} with data $\Z$; return the scores $\mathcal{V} = \{V_i\}_{i=1}^{|\mathcal{V}|}$, the score index $k+1$, and the ~~~~~~~~~~~optima $\hat{\wv}=(w_1, \cdots, w_{|\mathcal{V}|})$ from Step II of \Cref{alg:csa}
 \State Let $L=1$, $R=k+1$, M=$\lceil (L+R)/2 \rceil$, and initialization  $\wv_0=\hat{\wv}$
\State \textbf{while} ~$R-L>1$~\textbf{do} 
 \NoNumber{~~~~Solve \Cref{eq:qt-opt-simple} by a nonlinear programming solver (e.g.  \texttt{slsqp}($\cdot$) in \textsf{R}) with $J=M$, \\~~~~$n+1=|\mathcal{V}|$ and initialization $\wv_0$; return the optimal objective value $\hat{\alpha}_J$ and optima $\hat{\wv}_J$
 \State ~~~~\textbf{if} $\hat{\alpha}_J \geq \alpha$, $L \leftarrow M$ \textbf{else} $R \leftarrow M$;  $M \leftarrow \lceil (L+R)/2 \rceil$
  \State~~~~Set initialization $\wv_0 = \hat{\wv}_J$
 }
\end{algorithmic} 
 \textbf{Output:} Compute  $\hat{C}_t^{\Gamma}(X)$ by \Cref{eq:csa-interval} with $\hat{Q}(Z_{1:n},X) = V_{R}$ 

 \end{algorithm}
}

\section{Rejection Sampling for Counterfactual Outcome}
\label{sec:rejection-sampling}
We design an algorithm to generate $Y_i \sim p(Y(t) \g X_i,T_i=1-t)$ that is consistent with the MSM. Let the proposal distribution be $q(y \g x) = p(Y(t)=y \g X=x,T=t)$. The target distribution is 
\ba{
p(y\g x) = \eta(y\g x) q(y\g x) / M(x), 
\label{eq:tilt}
} 
where $ \eta(y \g x)$ is a positive tilting function and $M(x) = \int \eta(y \g x) q(y\g x) dy$ is the normalizing constant. The tilting function is chosen so that $ 1/\Gamma \leq \eta(y\g x) / M(x) \leq \Gamma$.  By construction, the density ratio $p(y\g x)/q(y\g x)$ is bounded between $1/\Gamma$ and $\Gamma$. Let the missing outcome distribution be $ p(Y(t)=y \g X=x,T=1-t) = p(y\g x) $.  By \Cref{lm:sscore}, the corresponding selection score belongs to $\cE(\Gamma)$.

Samples from $p(y\g x)$ can be generated by the rejection sampling. Since the density ratio is bounded by $ \Gamma$, we generate $Y(t) \sim q(y | x)$, and accept it with probability
\ba{
p_{\text{accept}} = \frac{p(Y(t) \g x)}{q(Y(t) \g x) \Gamma } = \frac{\eta(Y(t) \g x)}{\Gamma M(x)}. 
\label{eq:accept}
}
The rejection sampling guarantees the accepted samples are generated from $p(Y(t) \g X, T=1-t)$ with a sensitivity model belonging to $\cE(\Gamma)$. The choice of tilting function decides which sensitivity model that the true data generating distribution $\mathbb{P}_0$ corresponds to. By choosing different tilting functions, we can evaluate CSA over different $\mathbb{P}_0$ which are all compatible with the observed data and the MSM. 

For \Cref{fig:sample}, $q_{l}(x)$ and $q_{r}(x)$ are set to be the $1/(2(\Gamma+1))$ and $1-1/(2(\Gamma+1))$ quantile of the observed outcome distribution $p(Y(1) | X=x, T=1)$, respectively.  The tilting function $\eta(y | x)$  is set as $1/\Gamma$ when $q_{l}(x) \leq y \leq q_{r}(x)$, and as $\Gamma$ otherwise.
To generate samples of the potential outcome, the normalization constant $M(x)=1$ due to the tilting function design.  ~\looseness=-1

\section{Coverage Limit of a Finite-length Interval}
\label{sec:threshold}

We discuss a maximum coverage that a finite-length predictive interval can achieve in conformal inference.  For weighted conformal prediction, the width of the predictive interval is determined by the maximum $(1-\alpha)$-quantile of  the distribution, i.e.\ $ \sum_{i=1}^n p_i^{(h)}\delta_{V_i} + p_{n+1}^{(h)}\delta_{\infty}$. The maximum $(1-\alpha)$-quantile is infinite if there exist a sensitivity model $h \in \mathcal{H}(\lambda)$, such that the probability mass is greater than the miscoverage rate, i.e.\ $p_{n+1}^{(h)} \geq \alpha$.

To infer $Y(1)$, by \Cref{prop:qt-opt}, we get the minimal miscoverage $\alpha^*$ that an informative predictive interval can achieve
\ba{
\alpha^* \vcentcolon& = \max_{h \in \mathcal{H}(\lambda)} p_{n+1}^{(h)}(Z_{1:n}, X; \Gamma, t=1) \notag \\
&=  \frac{w_{hi}^{\Gamma,t=1}(X)}{\sum_{i=1}^n w_{lo}^{\Gamma,t=1}(X_i) + w_{hi}^{\Gamma,t=1}(X)} \notag \\
& =\frac{\Gamma(e(X)^{-1} -1)+1}{n+1+\sum_{i=1}^n(e(X_i)^{-1} -1)/\Gamma  + \Gamma(e(X)^{-1} -1)}.
\label{eq:alpha}
}

 The smaller the $\alpha^*$ is, the higher coverage an informative interval can reach. By setting $\Gamma = 1$ in \Cref{eq:alpha}, we get the minimal miscoverage $\alpha^*$ under unconfoundedness, and by further setting $e(X) = 1/2$, we get  $\alpha^*$ for  randomized trials.  These results are  in \Cref{tab:alpha}. 

\begin{table}[ht]
 \caption{  The minimal miscoverage for a finite-length predictive band. The target is $Y(t)$ of a new unit in population with covariates $X$, and the  observed data are  at treatment level $t$. 
 \label{tab:alpha}  }
 \centering
 \resizebox{0.9\columnwidth}{!}{%
\begin{tabular}{cccc}
\toprule
& 1:1 Randomized Trial  & Unconfoundedness & Confounding Specified by the MSM\\
\midrule 
$\alpha^*$ &   $\frac{1}{n+1}$   & $\frac{p(T=t  \g  X)^{-1}} {\sum_{i=1}^{n+1} p(T=t \g X_i)^{-1} }$ &  $\frac{w_{hi}^{\Gamma,t}(X)}{\sum_{i=1}^n w_{lo}^{\Gamma,t}(X_i) + w_{hi}^{\Gamma,t}(X)}$ \\
\bottomrule
\end{tabular}
}
\end{table}

\Cref{eq:alpha} reveals that to predict $Y(1)$ for a target unit,  $\alpha^*$ decreases when the number of observed units $n$ increases and when the sensitivity parameter $\Gamma$ decreases.  
$\alpha^*$ decreases when the propensity score $e(X_i)$ of the observed units decreases and when  $e(X)$ of the target point increases.
This change in propensity score indicates that there is overlap between the target unit's covariates and the covariates of the training units.

\section{Additional Simulation Details and Results}\label{sec:additional-sim}

The counterfactual distributions in \Cref{sec:exp-missing,sec:simulation-ite} are generated by the rejection sampling method in \Cref{sec:rejection-sampling}. Specifically, the middle figure in \Cref{fig:sample}, which we adopt as the adversarial case in the MSM, is generated by the tilting function in \Cref{eq:tilt} as $\eta(y \g x) = 1/\Gamma$ if $q_l(x) \leq y \leq q_r(x)$ and $\eta(y \g x) = \Gamma$ otherwise. Here, $q_l(x)$ and $q_r(x)$ are the $1/(2(\Gamma + 1))$ and $1-1/(2(\Gamma+1))$ quantiles of the observed outcome distribution, respectively. Such construction makes the normalizing constant $M(x) =1$ in \Cref{eq:tilt} and produces a counterfactual distribution with high density on the low density region of the observed distribution, which is an adversarial case to test the validity of the predictive band by CSA. In  \Cref{sec:simulation-ite}, the counterfactual distribution $p(Y(0) | X=x, T=1)$   is constructed by  the rejection sampling with the observed outcome distribution $p(Y(0) | X=x, T=0)$ and the aforementioned tilting function $\eta(y|x)$.

 \Cref{tab:coverage} shows the coverage of different interval estimates. The estimations that CSA is compared with are 
percentile bootstrap sensitivity analysis of the ATE (Boot.Sens.) \citep{zhao2019sensitivity},  the Bayesian Additive Regression Tree estimation of the CATE (BART) \citep{chipman2010bart}, and ITE estimation under unconfoundedness \citep{lei2020conformal}.  \Cref{fig:fish-2} shows the percentage of positive and negative intervals at different confounding strength and coverage rate. The predictive interval is constructed by the conformal quantile prediction. \Cref{fig:shrink-2} shows that CSA produces valid and relatively sharp predictive interval for the ITE under unmeasured confounding.

\begin{table}[t]
 \caption{ Empirical coverage of the ITE by different interval estimations. CSA is implemented with mean prediction. $\Gamma$ is the sensitivity parameter. For all methods, the target coverage is $0.8$. When there is unmeasured confounding,
CSA-M produces better coverage than alternative methods. }
 \centering
\begin{tabular}{p{2mm}>{\centering\arraybackslash}p{1cm}>{\centering\arraybackslash}p{2cm}>{\centering\arraybackslash}p{2cm}>{\centering\arraybackslash}p{2cm}>{\centering\arraybackslash}p{2cm}}
\toprule
&\multicolumn{1}{c}{$\Gamma$} & \multicolumn{1}{c}{CSA} & \multicolumn{1}{c}{Boot.Sens.}& \multicolumn{1}{c}{BART}& \multicolumn{1}{c}{ITE-NUC} \\
\midrule
 \parbox[t]{2mm}{\multirow{5}{*}{\rotatebox[origin=c]{90}{Homosc.}}} &   1.0 &  0.80 (0.02)  &  0.03 (0.00) &  0.62 (0.01) & 0.80 (0.02)\\
 &   1.5 &  0.83 (0.03)  &  0.16 (0.01) &  0.58 (0.01) &  0.75 (0.03)\\
 &   2.0 & 0.85 (0.02)   &  0.24 (0.01) &  0.54 (0.01)  &  0.69 (0.03)\\
 &   3.0 &  0.87 (0.04) &   0.32 (0.01) &  0.49 (0.01)  & 0.60 (0.03)\\
 &   4.0 &  0.87 (0.04) &  0.37 (0.01) &  0.46 (0.01)  & 0.54 (0.03)\\
 
 \midrule
  \parbox[t]{2mm}{\multirow{5}{*}{\rotatebox[origin=c]{90}{Heterosc.}}} &    1.0 & 0.80 (0.02)  &  0.02 (0.00) &  0.61 (0.02) &  0.81 (0.03)\\
 &  1.5 &  0.82 (0.03) &   0.16 (0.00) &  0.58 (0.01) &  0.74 (0.03)\\
 &  2.0 &  0.85 (0.03) &   0.25 (0.01) &  0.56 (0.01)&  0.70 (0.03)\\
 &  3.0 &  0.87 (0.03) &  0.36 (0.01) &  0.53 (0.02) & 0.64 (0.03) \\
 &  4.0 &  0.89 (0.03) &   0.41 (0.01) &  0.51 (0.01)  & 0.59 (0.04)\\
\bottomrule
\label{tab:coverage}
\end{tabular}
\end{table}

\begin{figure}[htbp]
\centering{
   {{\includegraphics[width=0.41\textwidth]{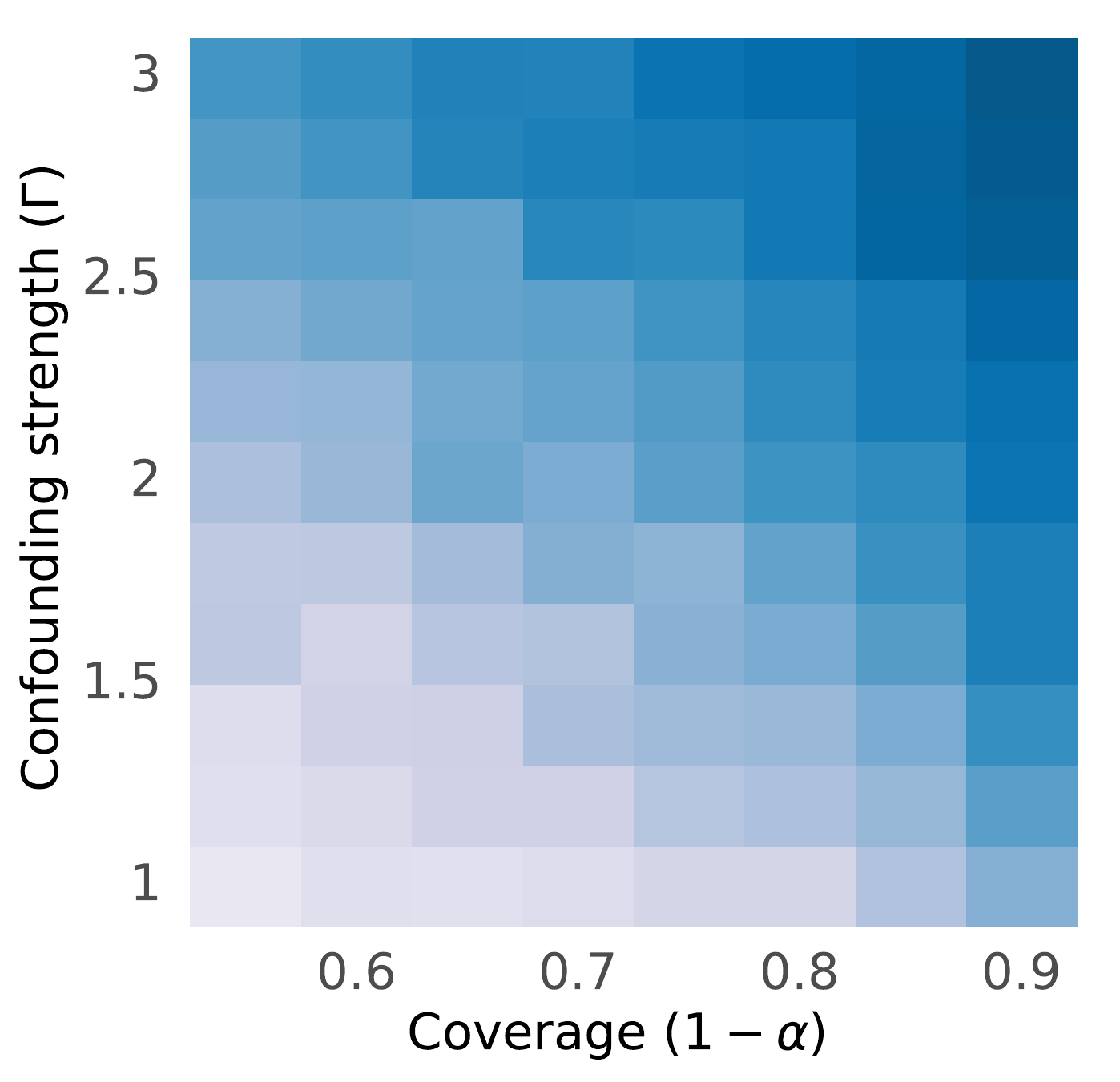} }}  
    {{\includegraphics[width=0.55\textwidth]{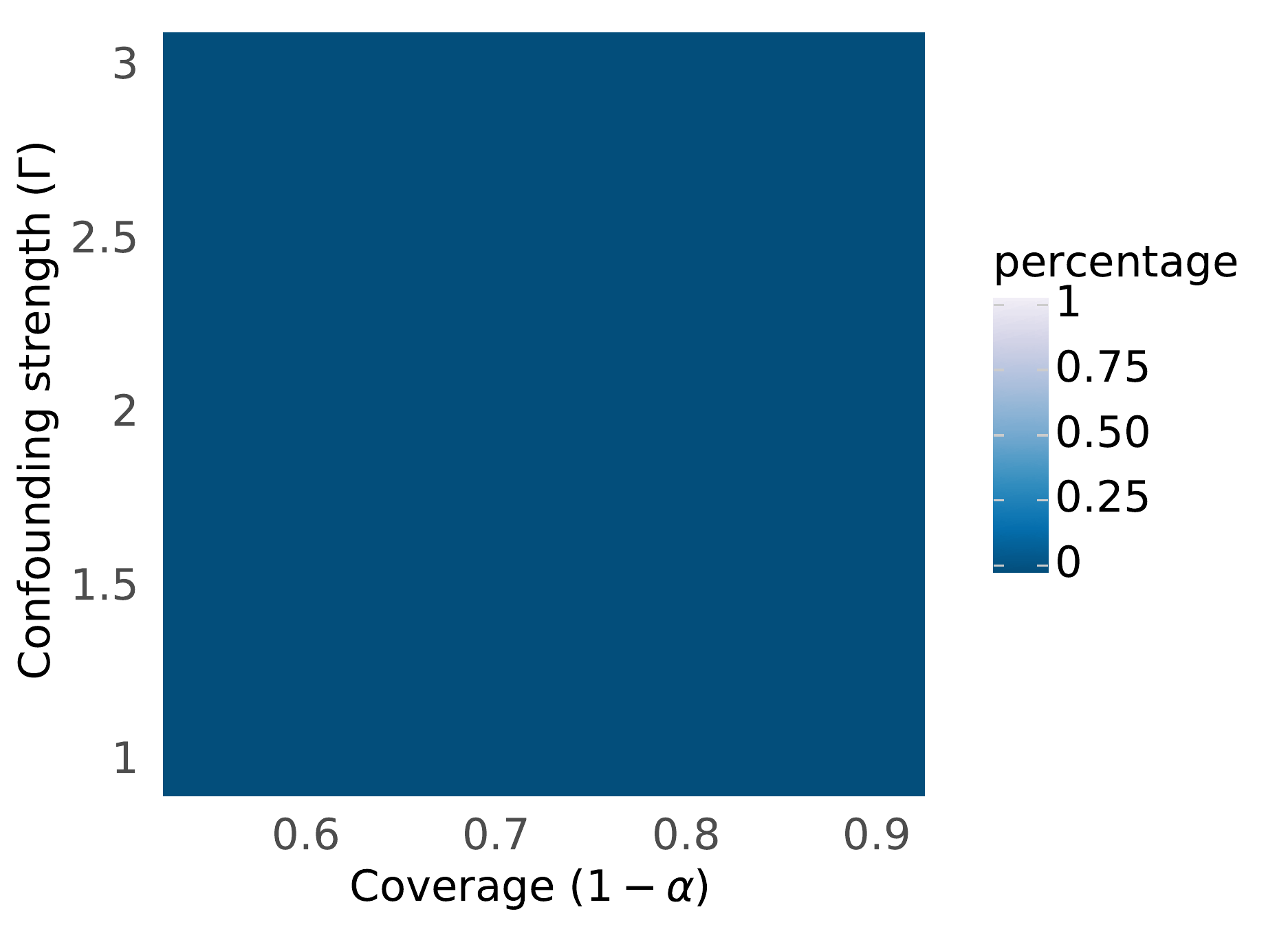} }} 
    \caption{Analysis of NHANES fish consumption data. The predictive interval is estimated by the CSA with quantile prediction. Left panel:  fraction of intervals with positive lower bounds; Right panel: fraction of intervals with negative upper bounds. }
    \label{fig:fish-2}%
 }
\end{figure}

\begin{figure}[ht]
\centering{
   {{\includegraphics[width=0.47\textwidth]{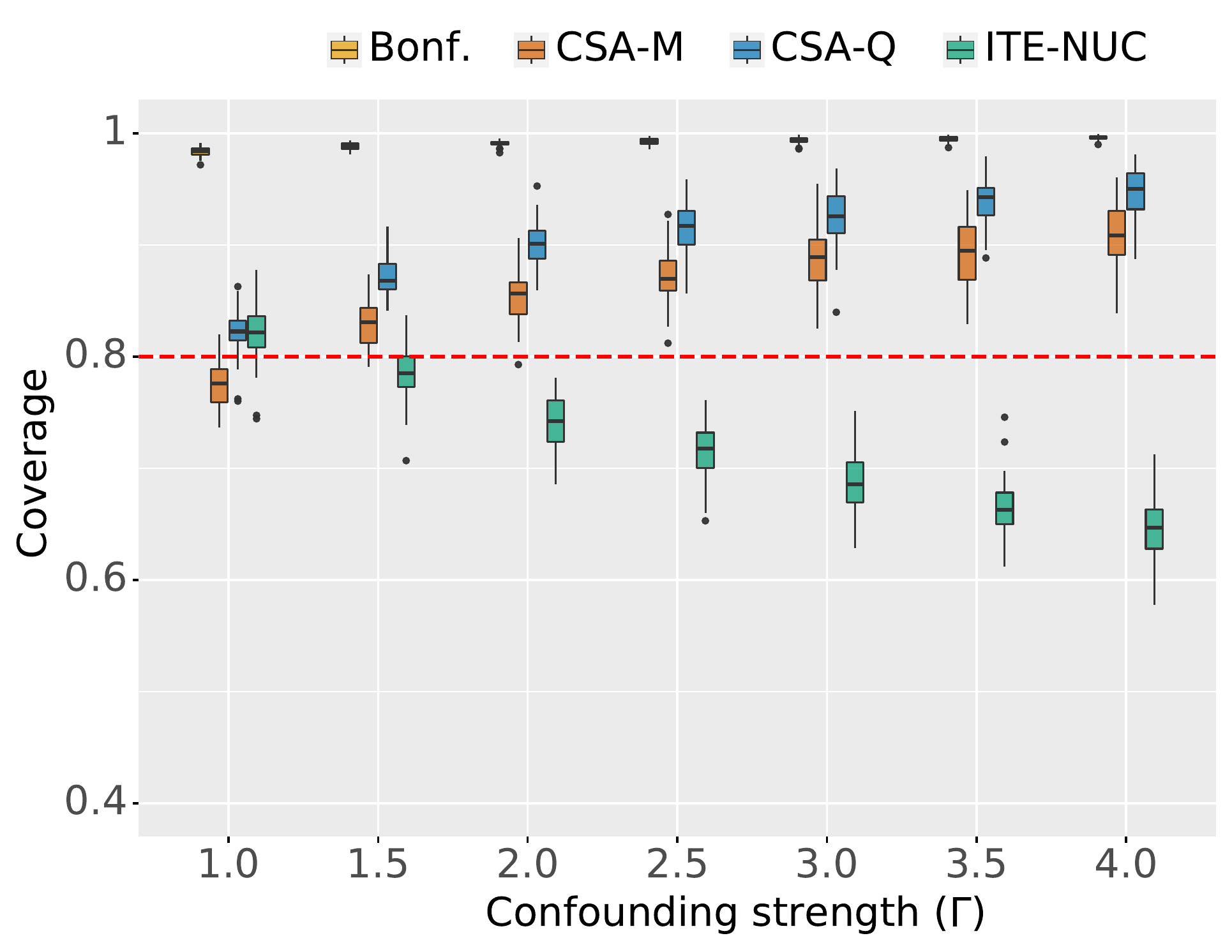} }} ~
   {{\includegraphics[width=0.47\textwidth]{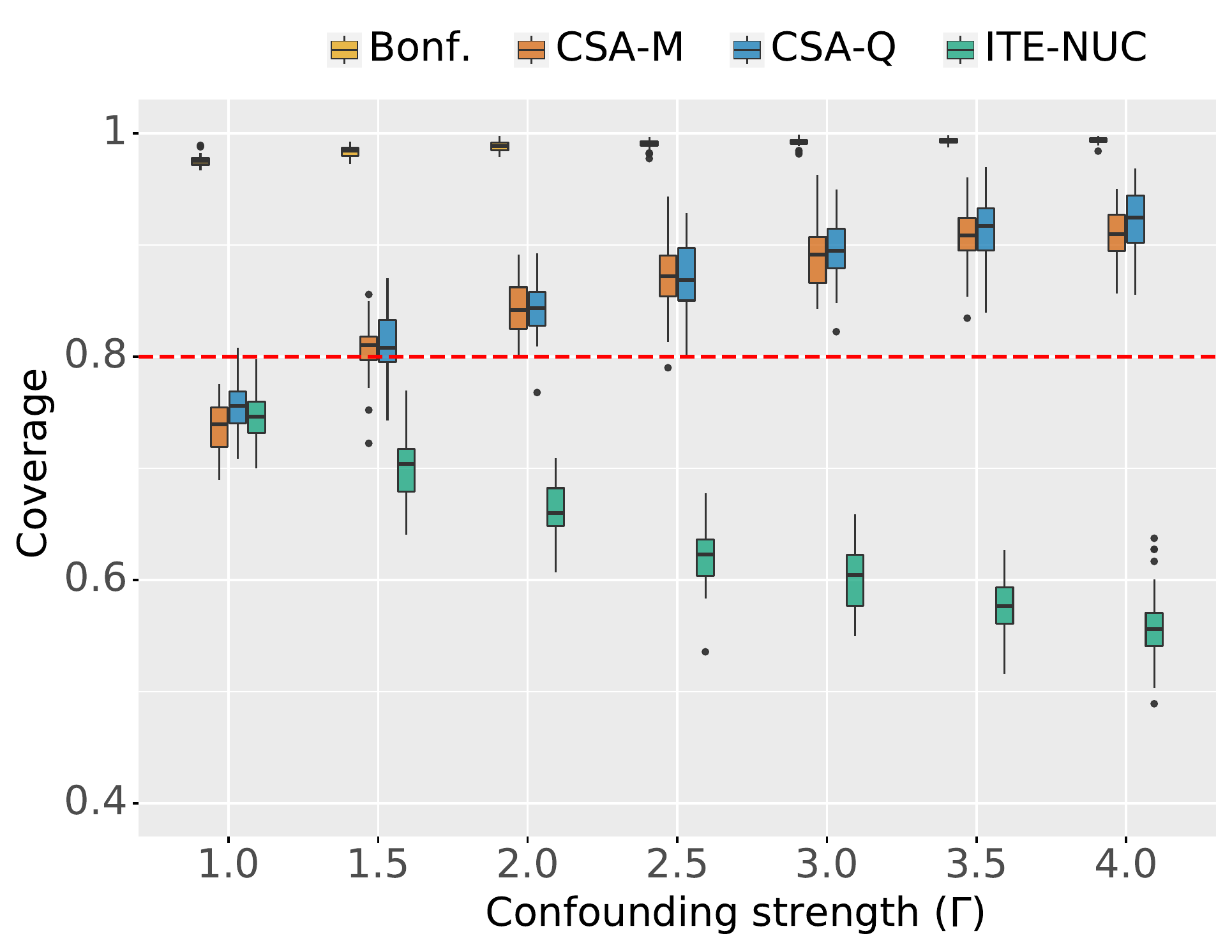} }} \\ \vspace{3mm}

   {{\includegraphics[width=0.47\textwidth]{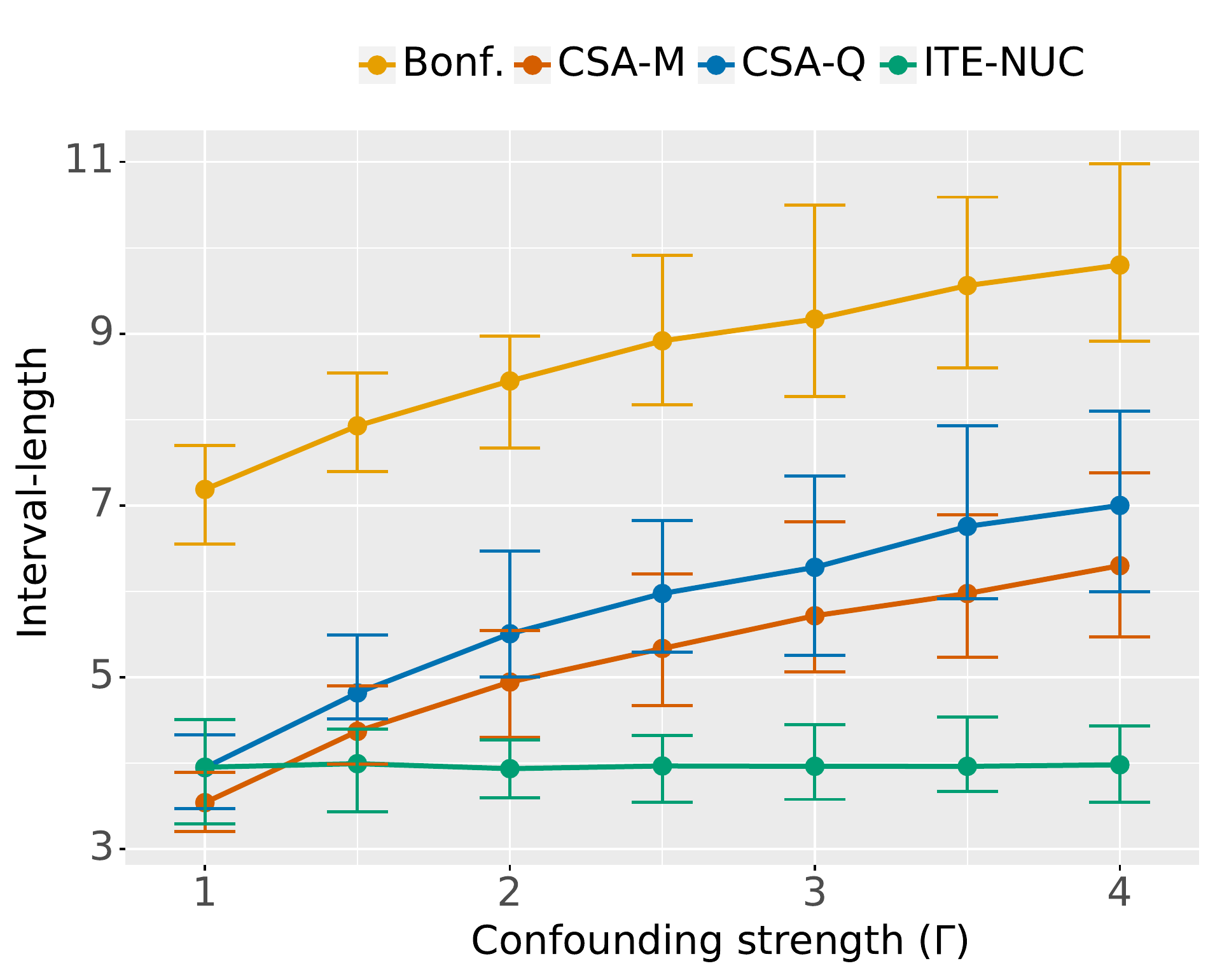} }} ~
   {{\includegraphics[width=0.47\textwidth]{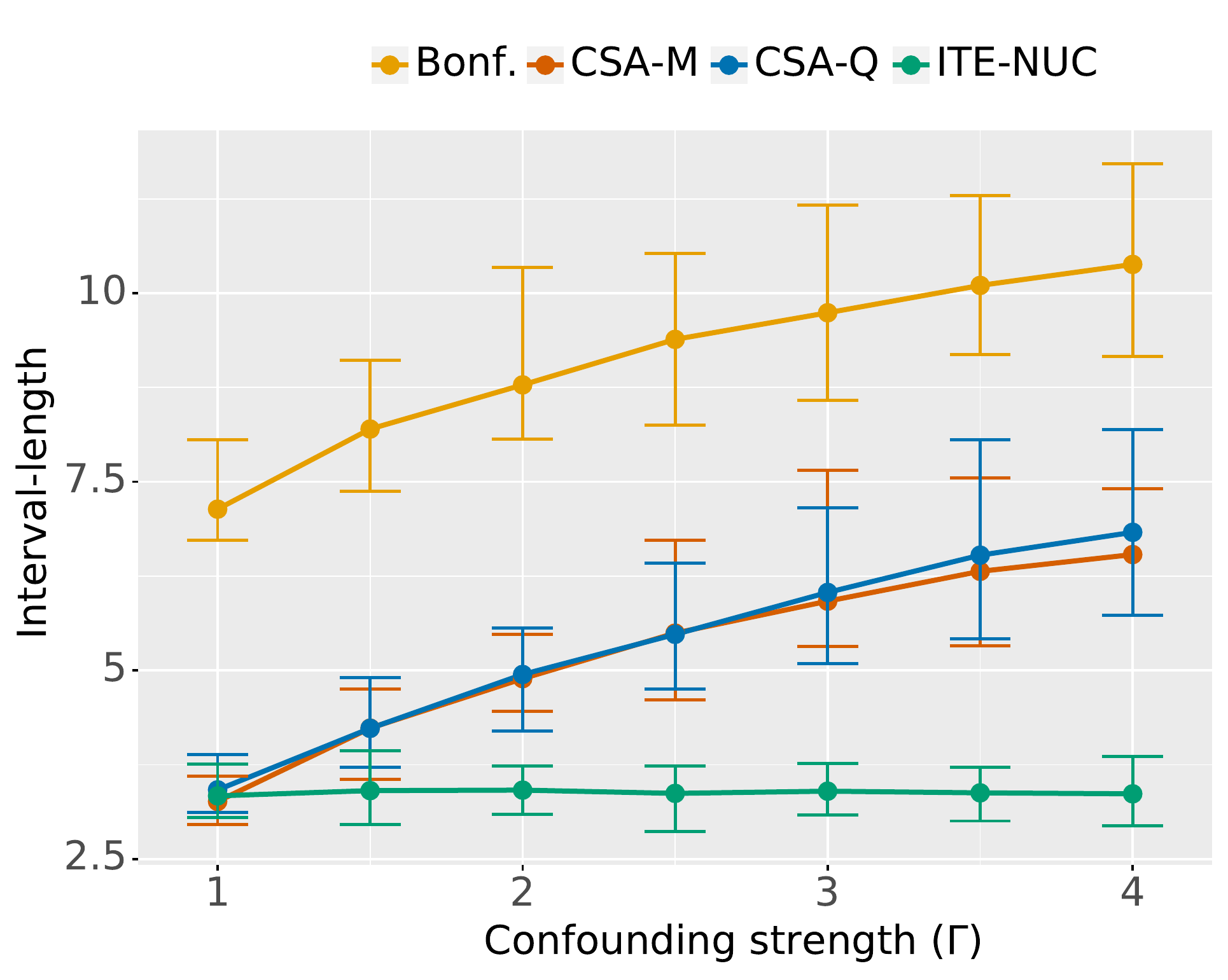} }}  \\
 \subfloat[Homosc.] 
  {{\includegraphics[width=0.47\textwidth]{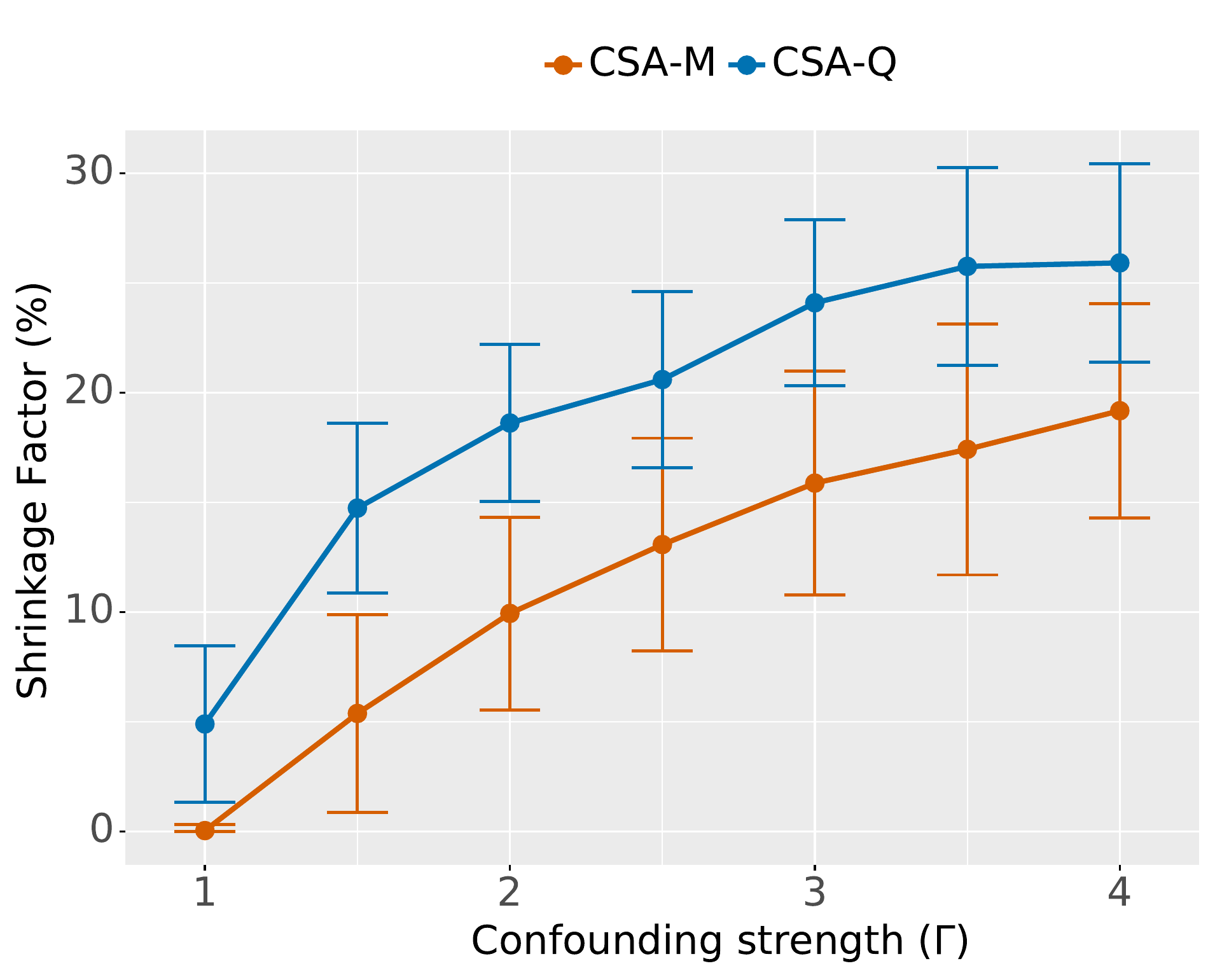}}} ~~~
   \subfloat[Heterosc.]
  {{\includegraphics[width=0.47\textwidth]{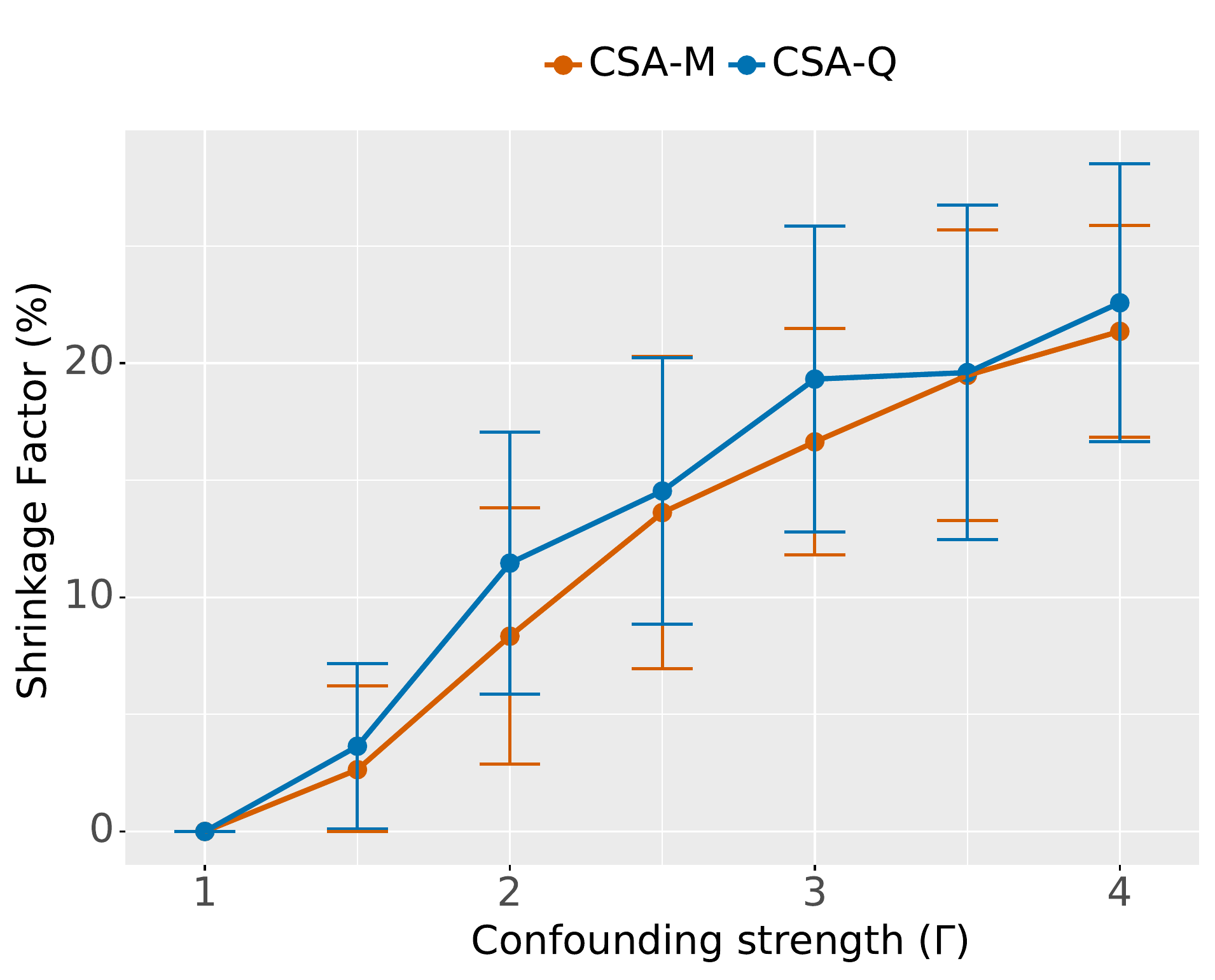} }} 
}
    \caption{\small \textbf{Top:} Empirical coverage of the ITE. The dashed lines is the target coverage.  \textbf{Middle:}  The length of predictive interval.   CSA with Bonferroni correction produces a wide interval and  high coverage. CSA with nested method has valid coverage and provides sharper intervals.
\textbf{Bottom:} The sharpness of  CSA predictions. CSA-M is less conservative than CSA-Q on homoscedastic data. To keep the nominal coverage, the maximal shrinkage factor is below 20\% for CSA-M and is slightly higher for CSA-Q. }
    \label{fig:shrink-2}
\end{figure}

\end{document}